\documentclass[12pt,letterpaper]{article}

\usepackage[ left=1in, top=1in, right=1in, bottom=1in]{geometry}
\usepackage{graphicx,bm,colonequals,amsmath,amssymb,url,xcolor,bbm}
\usepackage{array,tabularx,multirow}
\usepackage{enumitem}
\usepackage[font={footnotesize}]{caption,subcaption}
\usepackage[utf8]{inputenc}
\usepackage{enumitem}
\usepackage{soul} 
\usepackage{placeins} 
\usepackage[normalem]{ulem}
\usepackage{natbib,hyperref}

\graphicspath{{plots/}}

\usepackage{mathtools}
\mathtoolsset{showonlyrefs} 

\bibpunct[, ]{(}{)}{;}{a}{,}{,}
   
\usepackage{amsthm}
\newtheoremstyle{propstyle} 
    {2mm}                    
    {1mm}                    
    {\itshape}                   
    {}                           
    {\scshape}                   
    {.}                          
    {.5em}                       
    {}  
\theoremstyle{propstyle}
\newtheorem{proposition}{Proposition}
\theoremstyle{propstyle}

\theoremstyle{propstyle}
\newtheorem{lemma}{Lemma}
\newtheorem{corollary}{Corollary}

\theoremstyle{propstyle}

\theoremstyle{propstyle}

\usepackage[ruled,vlined]{algorithm2e}
\SetKwInput{KwInput}{Input}
\usepackage{algorithmic}
\usepackage{array,framed}
\usepackage{float}

\makeatletter
\renewcommand{\paragraph}{%
  \@startsection{paragraph}{4}%
  {\z@}{2ex \@plus 1ex \@minus .2ex}{-1em}%
  {\normalfont\normalsize\bfseries}%
}
\makeatother

\DeclareMathAlphabet\mathbfcal{OMS}{cmsy}{b}{n}

\newcommand{\bc}{\mathbf{c}}
\newcommand{\bd}{\mathbf{d}}
\renewcommand{\bf}{\mathbf{f}}
\newcommand{\bb}{\mathbf{b}}
\newcommand{\bs}{\mathbf{s}}

\newcommand{\by}{\mathbf{y}}

\newcommand{\bw}{\mathbf{w}}

\newcommand{\bz}{\mathbf{z}}

\newcommand{\bY}{\mathbf{Y}}

\newcommand{\bF}{\mathbf{F}}
\newcommand{\bG}{\mathbf{G}}

\newcommand{\bL}{\mathbf{L}}

\newcommand{\bI}{\mathbf{I}}
\newcommand{\bD}{\mathbf{D}}

\newcommand{\bK}{\mathbf{K}}

\newcommand{\bQ}{\mathbf{Q}}

\newcommand{\bC}{\mathbf{C}}

\newcommand{\bfzero}{\mathbf{0}}

\newcommand{\bfmu}{\bm{\mu}}

\newcommand{\bftheta}{\bm{\theta}}

\newcommand{\bfepsilon}{\bm{\epsilon}}

\newcommand{\bfSigma}{\bm{\Sigma}}

\newcommand{\bfLambda}{\bm{\Lambda}}

\newcommand{\diag}{diag}

\newcommand{\GP}{\mathcal{GP}}

\newcommand{\order}{\mathcal{O}}
\newcommand{\normal}{\mathcal{N}}
\newcommand{\nig}{\mathcal{NIG}}

\DeclareMathOperator*{\argmax}{arg\,max}

\newcommand{\domain}{\mathcal{D}}

\newcommand{\map}{\mathcal{T}}
\newcommand{\pmap}{\widetilde{\mathcal{T}}}

\title{Scalable Bayesian transport maps for high-dimensional non-Gaussian spatial fields}

\author{Matthias Katzfuss\thanks{Department of Statistics, Texas A\&M University. Corresponding author: \texttt{katzfuss@gmail.com}} \and Florian Sch{\"a}fer\thanks{School of Computational Science and Engineering, Georgia Institute of Technology}}

\date{}

\begin{document}

\maketitle

\begin{abstract}
A multivariate distribution can be described by a triangular transport map from the target distribution to a simple reference distribution. We propose Bayesian nonparametric inference on the transport map by modeling its components using Gaussian processes. This enables regularization and uncertainty quantification of the map estimation, while still resulting in a closed-form and invertible posterior map. 
We then focus on inferring the distribution of a nonstationary spatial field from a small number of replicates. We develop specific transport-map priors that are highly flexible and are motivated by the behavior of a large class of stochastic processes. Our approach is scalable to high-dimensional distributions due to data-dependent sparsity and parallel computations.
We also discuss extensions, including Dirichlet process mixtures for flexible marginals.
We present numerical results to demonstrate the accuracy, scalability, and usefulness of our methods, including statistical emulation of non-Gaussian climate-model output.
\end{abstract}

{\small\noindent\textbf{Keywords:} climate-model emulation; Dirichlet process mixture; Gaussian process; generative modeling; maximin ordering; nonstationarity}

\section{Introduction \label{sec:intro}}

\paragraph{Motivation}
Inference on a high-dimensional joint distribution based on a relatively small number of replicates is important in many applications.
For example, generative modeling of nonstationary and non-Gaussian spatial distributions is crucial for statistical climate-model emulation \citep[e.g.,][]{Castruccio2014,Nychka2018,Haugen2019}, in ensemble-based data assimilation \citep[e.g.,][]{Houtekamer2016,Katzfuss2015b}, and design studies for new satellite observing systems at NASA using observing system simulation experiments \citep{Errico2013}.

\begin{figure}
\centering
\includegraphics[trim=4mm 34mm 2mm 11mm, clip, page=2,width =.56\linewidth]{transportMapIllustration2}\\ 
\vspace{2mm}
\includegraphics[trim=2mm 60mm 4mm 7mm, clip, page=1,width =.6\linewidth]{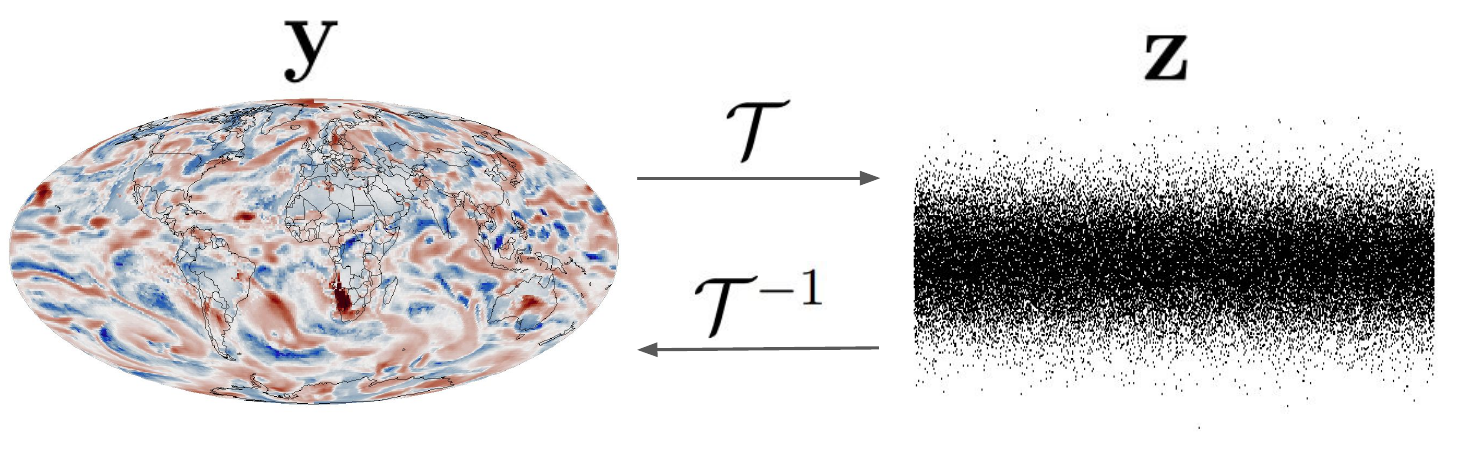} 
\caption{Top panel: Illustration of a transport map $\map$ transforming a (bivariate) non-Gaussian distribution $p(\by)$ to a standard Gaussian distribution $\normal(\bfzero,\bI)$. Bottom: Equivalently, $\map$ converts a realization (here, a spatial field) $\by \sim p(\by)$ to standard Gaussian coefficients $\bz = \map(\by) \sim \normal(\bfzero,\bI)$. Under maximin ordering (Figure \ref{fig:maxmin}), $\bz$ can be viewed as scores corresponding to a nonlinear version of principal components, and they decrease in importance and in corresponding spatial scale from left to right. The spatial field is output from a climate model on a grid of size $N = 288 \times 192 = 55{,}296$; we want to learn $\map$ characterizing the $N$-dimensional distribution based on an ensemble of $n<100$ training samples (see Section \ref{sec:application}).}
\label{fig:illus}
\end{figure}


\paragraph{Transport maps}
A continuous multivariate distribution with any dependence structure can be characterized via a triangular transport map \citep[see][for a review]{Marzouk2016} that transforms the target distribution to a reference distribution (e.g., standard Gaussian), as illustrated in Figure \ref{fig:illus}. For Gaussian target distributions, such a map is linear and given by the Cholesky factor of the precision matrix; non-Gaussian distributions can be obtained by allowing nonlinearities in the map. Given an invertible transport map, it is straightforward to sample from the target distribution and some of its conditionals, or to transform the non-Gaussian data to the reference space, in which simple linear operations such as regression or interpolation can be applied.
Typically, the map is estimated based on training data, often by iteratively expanding a finite-dimensional parameterization of the transport map \citep[e.g.,][]{ElMoselhy2012,Bigoni2016,Marzouk2016,baptista2020adaptive}; subsequent inference is then carried out assuming that the map is known.

\paragraph{Bayesian transport maps}
We propose an approach for Bayesian inference on a transport map that describes a multivariate continuous distribution and is learned from a limited number of samples from the distribution.
We model the map components using nonparametric, conjugate Gaussian-process priors, which probabilistically regularize the map and shrink toward linearity. The resulting generative model is flexible, naturally quantifies uncertainty, and adjusts to the amount of complexity that is discernible from the training data, thus avoiding both over- and under-fitting. The conjugacy results in simple, closed-form inference. 
Instead of assuming Gaussianity for the multivariate target distribution, our approach is equivalent to a series of conditional GP regression problems that together characterize a non-Gaussian target distribution.

\paragraph{Transport maps for spatial fields}
We then focus on learning or emulating structured target distributions corresponding to spatial fields observed at a finite but large number of locations, based on a relatively small number of training replicates.
In this setting, our Bayesian transport maps impose sparsity and regularization motivated by the behavior of diffusion-type processes that are encountered in many environmental applications.
After applying a so-called maximin ordering of the spatial locations, determining the triangular transport map essentially consists of conditional spatial-prediction problems on an increasingly fine scale. We discuss how this scale decay results in conditional near-Gaussianity for a large class of non-Gaussian stochastic processes associated with quasilinear partial differential equations. Hence, our prior distributions are motivated by the behavior of Gaussian fields with Mat\'ern-type covariance, for which the so-called screening effect leads to a decay of influence that motivates sparse transport maps that only consider nearby observations in the spatial prediction problems, corresponding to assumptions of conditional independence. The degree of shrinkage and sparsity are determined by hyperparameters that are inferred from data.
The resulting Bayesian methods require little user input, scale near-linearly in the number of spatial locations, and the main computations are trivially parallel.

\paragraph{Extensions}
We further increase the flexibility in the (continuous) marginal distributions by modeling the GP-regression error terms using Dirichlet process mixtures, which can be fit using a Gibbs sampler. The resulting method lets the data decide the degrees of nonlinearity, nonstationarity, and non-Gaussianity, without manual tuning or model-selection.
We also discuss an extension for settings in which Euclidean distance between the locations is not meaningful or in which variables are not identified by spatial locations (e.g., multivariate spatial processes).

\paragraph{Related existing spatial methods}
Most existing methods for spatial inference are in principle applicable in our emulation setting, but they are often geared toward spatial prediction based on a single training replicate and assume Gaussian processes (GPs) with simple parametric covariance functions \citep[e.g.,][]{Cressie1993,Banerjee2004}. Many extensions to nonstationary \citep[e.g., as reviewed by][]{Risser2016} or nonparametric covariances \citep[e.g.,][]{Huang2011,Choi2013,Porcu2019} have been proposed, but these typically still rely on implicit or explicit assumptions of Gaussianity.
This includes locally parametric methods specifically developed for climate-model emulation \citep{Nychka2018,Wiens2020,Wiens2021} that locally fit anisotropic Mat\'ern covariances in small windows and then combine the local fits into a global model.
For non-Gaussian spatial data, GPs can be transformed or used as latent building blocks \citep[see, e.g.,][and references therein]{Gelfand2016,Xu2017}, but relying on a GP's covariance function limits the types of dependence that can be captured. Parametric non-Gaussian Mat\'ern fields can be constructed using stochastic partial differential equations driven by non-Gaussian noise \citep{Wallin2015,Bolin2020}.
Models for non-Gaussian spatial data can also be built using copulas; for example, \citet{Graler2014} proposed vine copulas for spatial fields with extremal behavior, and the factor copula approach of \citet{Krupskii2018} assumes all locations in a homogeneous spatial region to be affected by a common latent factor.
Many existing non-Gaussian spatial methods are not scalable to large datasets.

\paragraph{Vecchia and extensions}
A popular way to achieve scalability for Gaussian spatial fields with parametric covariances is via the Vecchia approximation \citep[e.g.,][]{Vecchia1988,Stein2004,Datta2016,Katzfuss2017a,Schafer2020}, which implicitly utilizes a linear transport map given by a sparse inverse Cholesky factor. \citet{Kidd2020} proposed a Bayesian approach to infer the Cholesky factor nonparametrically. Our (sparse) nonlinear transport maps can be viewed as a Bayesian, nonparametric, and non-Gaussian generalization of Vecchia approximations.

\paragraph{Generative models in machine learning}
A close relative of transport maps in machine learning are normalizing flows \citep[see][for a review]{kobyzev2020normalizing}, where triangular layers are used to ensure easy evaluation and inversion of likelihood objectives.
Variational autoencoders (VAEs) and generative adversarial networks (GANs) relying on deep neural networks \citep[e.g.,][]{Goodfellow2016} can be highly expressive and have been used for climate-model emulation \citep[e.g.,][]{Ayala2021,Besombes2021}.
\citet{Kovachki2020} designed GANs with triangular generators that allow conditional sampling.
Our approach can be viewed as a Bayesian shallow autoencoder, with the posterior transport map and its inverse acting as the encoder and decoder, respectively. 
In contrast to our proposed method, deep-learning approaches typically require massive training data, can be expensive to train, and are often highly sensitive to tuning-parameter and network-architecture choices \citep[e.g.,][]{Arjovsky2017,Hestness2017,Mescheder2018}.
Hence, in many low-data applications such approaches are only useful when paired with laborious and application-specific techniques, such as data augmentation, transfer learning, or advances in physics-informed machine learning \citep[e.g.,][]{Kashinath2021}.

\paragraph{Outline}
In Section \ref{sec:bayesian}, we develop Bayesian transport maps. In Section \ref{sec:spatial}, we consider the special case of high-dimensional spatial distributions. In Section \ref{sec:dpm}, we discuss extensions to non-Gaussian errors using Dirichlet process mixtures. Sections \ref{sec:simstudy} and \ref{sec:application} provide comparisons and applications to simulated data and climate-model output, respectively. Section \ref{sec:conclusions} concludes and discusses future work. Appendices \ref{app:proofs}--\ref{app:vae} contain proofs and further details. Fully automated implementations of our methods, along with code to reproduce all results, are available at \url{https://github.com/katzfuss-group/BaTraMaSpa}.

\section{Bayesian transport maps\label{sec:bayesian}}

\subsection{Transport maps and regression\label{sec:regression}}

Consider a continuous random vector $\by = (y_1,\ldots,y_N)^\top$, for example describing a spatial field at $N$ locations as in Figure \ref{fig:precdata}. For simplicity, assume that $\by$ has been centered to have mean zero.

For a multivariate Gaussian distribution, $\by \sim \normal_N(\bfzero,\bfSigma)$ with $\bfSigma^{-1} = \bL^\top\bL$, the (lower-triangular) Cholesky factor $\bL$ represents a transformation to a standard normal: $\bz = \bL \by \sim \normal_N(\bfzero,\bI_N)$. As a natural extension, we can characterize any continuous $N$-variate distribution $p(\by)$ by a potentially nonlinear transport map $\map: \mathbb{R}^N \rightarrow \mathbb{R}^N$ \citep[][]{Villani2009}, such that $\bz = \map(\by) \sim \normal_N(\bfzero,\bI_N)$ for $\by \sim p(\by)$. 
Like $\bL$, we can assume without loss of generality that the transport map $\map$ is lower-triangular \citep{Rosenblatt1952,Carlier2009},
\begin{equation}
\map(\by) = \begin{bmatrix*}[l] \map_1(y_1) \\ \map_2(y_1,y_2) \\ ~\,\vdots \\ \map_N(y_1,y_2,\ldots,y_N) \end{bmatrix*}, \label{eq:map}
\end{equation}
where each $\map_i(\by_{1:i})$ with $\by_{1:i} = (y_1,\ldots,y_i)^\top$ is an increasing function of its $i$th argument to ensure that $\map$ is invertible and implies a proper density $p(\by)$.
Letting $\normal(x|\mu,\sigma^2)$ denote a Gaussian density with parameters $\mu$ and $\sigma^2$ evaluated at $x$, we then have
\begin{equation}
    \label{eq:mapdist}
\textstyle p(\by) = p_\bz\big(\map(\by)\big) \, |\text{det} \nabla \map| = \prod_{i=1}^N \big( \, \normal(\map_i(\by_{1:i})|0,1) \, \big|\frac{\partial \map_i(\by_{1:i})}{\partial y_i}\big| \, \big),
\end{equation}
as the triangular $\map$ also implies a triangular $\nabla \map = ( \frac{\partial \map_i(\by_{1:i})}{\partial y_j} )_{i,j=1,\ldots,N}$.

Throughout, we assume each $\map_i$ to be linearly additive in its $i$th argument,
\begin{equation}
    \label{eq:condmap}
\map_i(\by_{1:i}) = (y_i - f_i(\by_{1:i-1}))/d_i, \qquad i=1,\ldots,N,
\end{equation}
for some $d_i \in \mathbb{R}^+$, $f_i: \mathbb{R}^{i-1} \rightarrow \mathbb{R}$ for $i=2,\ldots,N$, and $f_i(\by_{1:i-1}) \equiv 0$ for $i=1$.
Then, $\partial_i \map_i(\by_{1:i}) = 1/d_i > 0$, as required.
Using \eqref{eq:mapdist}, it is easy to show that
\begin{equation}
\label{eq:map2norm}
   \textstyle p(\by) \propto  \prod_{i=1}^N \big( \exp(-\frac{1}{2d_i^2}(y_i - f_i(\by_{1:i-1}))^2) \, \frac{1}{d_i} \big) \propto \prod_{i=1}^N \normal(y_i|f_i(\by_{1:i-1}),d_i^2).
\end{equation}
Thus, the transport-map approach has turned the difficult problem of inferring the $N$-variate distribution of $\by$ into $N$ independent regressions of $y_i$ on $\by_{1:i-1}$ of the form
\begin{equation}
y_i = f_i(\by_{1:i-1}) + \epsilon_i, \quad \epsilon_i \sim \normal(0,d_i^2), \qquad i=1,\ldots,N.  \label{eq:regi}
\end{equation}

Sparsity in the map $\map$ corresponds to conditional independence in the joint distribution $p(\by)$ \citep[cf.][]{Spantini2018}. Specifically, if we assume $f_i(\by_{1:i-1}) = f_i(\by_{c_i})$ for a subset $c_i \subset \{1,\ldots,i-1\}$, then $\map$ is sparse in that $\map_i$ only depends on $y_j$ if $j \in c_i$ (or if $j=i$). Making such a sparsity assumption for $i=2,\ldots,N$ (and setting $\by_{c_1} = \emptyset$), we have from \eqref{eq:map2norm} that $p(\by) = \prod_{i=1}^N p(y_i|\by_{c_i})$, meaning that $y_i$ is independent of $\{y_j: j \notin c_i, j<i\}$ conditional on $\by_{c_i}$. 
We will exploit this sparsity for computational gain for inferring large non-Gaussian spatial fields in Section \ref{sec:spatial}.

\subsection{Modeling the map functions using Gaussian processes\label{sec:priors}}

In the existing transport-map literature \citep[e.g.,][]{Marzouk2016}, $f_i: \mathbb{R}^{i-1} \rightarrow \mathbb{R}$ and $d_i \in \mathbb{R}^+$ in \eqref{eq:condmap}--\eqref{eq:regi} are often assumed to have parametric form, whose parameters are estimated and then assumed known. 
Instead, we here assume a flexible, nonparametric prior on the map $\map$ by specifying independent conjugate Gaussian-process-inverse-Gamma priors for the $f_i$ and $d_i^2$. These prior assumptions induce prior distributions on the map components $\map_i$ in \eqref{eq:condmap}, and thus on the entire map $\map$ in \eqref{eq:map}.

Specifically, for the ``noise'' variances $d_i^2$, we assume inverse-Gamma distributions, 
\begin{equation}
\label{eq:dprior}
    d_i^2 \stackrel{ind.}{\sim} \mathcal{IG}(\alpha_i,\beta_i), \qquad \text{with } \alpha_i>1, \; \beta_i > 0, \qquad i=1,\ldots,N.
\end{equation}    
Conditional on $d_i^2$, each function $f_i$ is modeled as a Gaussian process (GP) with inputs $\by_{1:i-1}$,
\begin{equation}
    \label{eq:gp}
f_i | d_i \stackrel{ind.}{\sim} \GP(0,d_i^2 K_i), \qquad i=1,\ldots,N,
\end{equation}
where $K_i(\cdot,\cdot) = C_i(\cdot,\cdot)/E(d_i^2)$, $E(d_i^2) = \beta_i/(\alpha_i -1)$,
\begin{equation}
    \label{eq:kernel}
C_i(\by_{1:i-1},\by_{1:i-1}') = \by_{1:i-1}^\top \bQ_i \by_{1:i-1}' + \sigma^2_i \, \rho_i(\by_{1:i-1},\by_{1:i-1}'), \qquad i=1,\ldots,N,
\end{equation}
$\sigma_i \in \mathbb{R}^+_0$, and $\rho_i$ is a positive-definite correlation function such that $\rho_i(\by_{1:i-1},\by_{1:i-1}) =1$.
This prior on $f_i$ is motivated by considering $\tilde f_i|\bb_i \sim \GP(\bb_i^\top(\cdot),\sigma^2_i \rho_i(\cdot,\cdot))$ with inputs $\by_{1:i-1}$, where $\bb_i \sim \normal(\bfzero,\bQ_i)$. Integrating out $\bb_i$, we obtain $\tilde f_i \sim \GP(0,C_i)$ with $C_i$ as in \eqref{eq:kernel}, and hence $f_i = (d_i/\sqrt{E(d_i^2)})\tilde f_i$ as in \eqref{eq:gp}.
The degree of nonlinearity of $f_i$ is determined by $\sigma_i^2$; if $\sigma_i^2=0$, then $f_i$ is a linear function of $\by_{1:i-1}$.
The prior distributions (i.e., $\alpha_i$, $\beta_i$, $C_i$) may depend on hyperparameters $\bftheta$; see Section \ref{sec:hyper} for more details.

\subsection{The posterior map\label{sec:invertiblemap}}

Now assume that we have observed $n$ independent training samples $\by^{(1)},\ldots,\by^{(n)}$ from the distribution in Section \ref{sec:regression} conditional on $\bf=(f_1,\ldots,f_N)$ and $\bd=(d_1,\ldots,d_N)$, such that 
$\by^{(j)} \stackrel{i.i.d.}{\sim} p(\by | \bf,\bd)$ with 
$\map(\by^{(j)}) \, | \, \bf,\bd \, \sim \normal_N(\bfzero,\bI_N)$, $j=1,\ldots,n$.
We combine the samples into an $n \times N$ data matrix $\bY$ whose $j$th row is given by $\by^{(j)}$. Then, for the regression in \eqref{eq:regi}, the responses $\by_i$ and the covariates $\bY_{1:i-1}$ are given by the $i$th and the first $i-1$ columns of $\bY$, respectively. Below, let $\by^\star$ denote a new observation sampled from the same distribution, $\by^\star \sim p(\by | \bf,\bd)$, independently of $\bY$.

Based on the prior distribution for $\bf$ and $\bd$ in Section \ref{sec:priors}, we can now determine the posterior map $\pmap$ learned from the training data $\bY$, with $\bf$ and $\bd$ integrated out. This map is available in closed form and invertible:
\begin{proposition}
\label{prop:maps}
The transport map $\pmap$ from $\by^\star \sim p(\by|\bY)$ to 
$\bz^\star = \pmap(\by^\star) \sim \normal_N(\bfzero,\bI_N)$ is a triangular map with components
\begin{equation}
z_i^\star = \pmap_i(y_1^\star,\ldots,y_i^\star) = \Phi^{-1}\big( F_{2\tilde\alpha_i}\big( \hat d_i^{-1} (v_i(\by^\star_{1:i-1})+1)^{-1/2}(y_i^\star - \hat f_i(\by^\star_{1:i-1})) \big)\big), \quad i=1,\ldots,N, \label{eq:singlemap}
\end{equation}
where
$\tilde\alpha_i = \alpha_i + n/2$, 
$\tilde\beta_i = \beta_i + \by_i{}^\top \bG_i^{-1} \by_i/2$,
$\hat d_i^2=\tilde\beta_i/\tilde\alpha_i$,
$\bG_i = \bK_i + \bI_n$, $\bK_i= K_i(\bY_{1:i-1},\bY_{1:i-1}) =\big(K_i(\by_{1:i-1}^{(j)},\by_{1:i-1}^{(l)}) \big)_{j,l=1,\ldots,n}$,
\begin{align}
\hat f_i(\by^\star_{1:i-1}) & = K_i(\by^\star_{1:i-1},\bY_{1:i-1})\bG_i^{-1}\by_i, \label{eq:gppredmean}\\ 
v_i(\by^\star_{1:i-1}) & = K_i(\by^\star_{1:i-1},\by^\star_{1:i-1}) - K_i(\by^\star_{1:i-1},\bY_{1:i-1})\bG_i^{-1} K_i(\bY_{1:i-1},\by^\star_{1:i-1}), \label{eq:gppredvar}
\end{align}
for $i=2,\ldots,N$, $\hat f_1 = v_1 = 0$ for $i=1$, and $\Phi$ and $F_{\kappa}$ denote the cumulative distribution functions of the standard normal and the $t$ distribution with $\kappa$ degrees of freedom, respectively.
The inverse map $\pmap^{-1}$ can be evaluated at a given $\bz^\star$ by solving the nonlinear triangular system $\pmap(\by^\star) =\bz^\star$ for $\by^\star$; because $\pmap$ is triangular, the solution can be expressed recursively as:
\begin{equation}
y_i^\star = \hat f_i(\by_{1:i-1}^\star) + F_{2\tilde\alpha_i}^{-1}(\Phi(z_i^\star))\, \hat d_i (v_i(\by_{1:i-1}^\star)+1)^{1/2}, \quad i=1,\ldots,N. \label{eq:invmap}
\end{equation}
\end{proposition}
All proofs are provided in Appendix \ref{app:proofs}. 
We can write the prior map in a similar form, but this is only useful in the case of highly informative priors.

Determining $\pmap_i$ requires $\order(n^3 + i n^2)$ time, mostly for computing and decomposing the $n \times n$ matrix $\bG_i$, for each $i=1,\ldots,N$. However, note that the $N$ rows or components of $\pmap$ can be computed completely in parallel, as in the optimization-based transport-map estimation reviewed in \citet{Marzouk2016}. Each application of the transport map or its inverse then consists of the GP prediction in \eqref{eq:gppredmean}--\eqref{eq:gppredvar} and only requires $\order(n^2 + in)$ time for $i=1,\ldots,N$, but the inverse map is evaluated recursively (i.e., not in parallel).

In contrast to existing transport-map approaches, our approach is Bayesian and naturally quantifies uncertainty in the nonlinear transport functions. The GP priors on the $f_i$ automatically adapt to the amount of information available, only resulting in strongly nonlinear function estimates when supplied the requisite evidence by the data.
If $n$ is increasing, then $\tilde\alpha_i$ increases, $F_{2\tilde\alpha_i}$ converges to $\Phi$, and $v_i(\by^\star_{1:i-1})$ typically converges to zero, and so the map components simplify to
\begin{equation}
\label{eq:detmap}
\pmap_i(y_1^\star,\ldots,y_i^\star) = (y_i^\star - \hat f_i(\by_{1:i-1}^\star))/\hat d_i  \quad \text{and} \quad  y_i^\star = \hat f_i(\by_{1:i-1}^\star) + z_i \hat d_i.
\end{equation}
When employed for finite $n$, this simplified version of the map ignores posterior uncertainty in $\bf$ and $\bd$ and instead relies on the point estimates $\hat f_i(\by_{1:i-1})$ and $\hat d_i^2$. 
If we further assume that $\sigma_i=0$ in \eqref{eq:kernel} for all $i=1,\ldots,N$, then all $f_i$ and all $\pmap_i$ become linear functions; we can think of the resulting linear map $\pmap(\by^\star) = \bL^\top \by^\star$ as an inverse Cholesky factor, in the sense that $\by^\star |\bY \sim \normal(\bfzero,\bfLambda^{-1})$ with $\bfLambda = \bL\bL^\top$.

Transport maps can be used for a variety of purposes. For example, we can obtain new samples $\by^\star$ from the posterior predictive distribution $p(\by|\bY)$ by sampling $\bz^\star \sim \normal_N(\bfzero,\bI_N)$ and computing $\by^\star=\pmap^{-1}(\bz^\star)$ using \eqref{eq:invmap}.
The map $\pmap$ in \eqref{eq:singlemap} provides a transformation from a non-Gaussian vector $\by^\star$ to the standard Gaussian $\bz^\star = \pmap(\by^\star)$; we call $\bz^\star=(z_1^\star,\ldots,z_N^\star)$ the map coefficients corresponding to $\by^\star$ (see Figure \ref{fig:illus} for an illustration).
Because the nonlinear dependencies have been removed, many operations are more meaningful on $\bz^\star$ than on $\by^\star$, including linear regressions, translations using linear shifts, and quantifying similarity using inner products. We can also detect inadequacies of the map $\pmap$ for describing the target distribution by examining the degree of non-Gaussianity and dependence in $\bz^\star$.
These uses of transport maps will be considered further in Section \ref{sec:spatialinference}.

\subsection{Hyperparameters\label{sec:hyper}}

The prior distributions on the $f_i$ and $d_i$ in Section \ref{sec:priors} 
may depend on unknown hyperparameters $\bftheta$. For example, by making inference on hyperparameters in the $\sigma_i$ in \eqref{eq:kernel}, we can let the data decide the degree of nonlinearity in the map and thus the non-Gaussianity in the resulting joint target distribution.
We can write in closed form the integrated likelihood $p(\bY)$, where $\bf$ and $\bd$ have been integrated out. 

\begin{proposition}
\label{prop:lik}
The integrated likelihood is
\begin{equation}
    p(\bY) \textstyle \propto \prod_{i=1}^N \big( \, |\bG_i|^{-1/2} \times ({\beta_i^{\alpha_i}}/{\tilde\beta_i^{\tilde\alpha_i}}) \times {\Gamma(\tilde\alpha_i)}/{\Gamma(\alpha_i)} \, \big),
    \label{eq:intlik}
\end{equation}
where $\Gamma(\cdot)$ denotes the gamma function, and $\tilde\alpha_i$, $\tilde\beta_i$, $\bG_i$ are defined in Proposition \ref{prop:maps}.
\end{proposition}

Now denote by $p_{\bftheta}(\bY)$ the integrated likelihood $p(\bY)$ computed based on a particular value $\bftheta$ of the hyperparameters. There are two main possibilities for inference on $\bftheta$.
First, an empirical Bayesian approach consists of estimating $\bftheta$ by the value that maximizes $\log p_{\bftheta}(\bY)$, and then regarding $\bftheta$ as fixed and known. As $\log p_{\bftheta}(\bY)$ is a sum of $N$ simple terms, it is straightforward to optimize this function using stochastic gradient ascent based on automatic differentiation.
Second, we can carry out fully Bayesian inference by specifying a prior $p(\bftheta)$, and sampling $\bftheta$ from its posterior distribution $p(\bftheta|\bY) \propto p_{\bftheta}(\bY) p(\bftheta)$ using Metropolis-Hastings; subsequent inference then relies on these posterior draws.

For our numerical results, we employed the empirical Bayesian approach, because it is faster and preserves the closed-form map properties in Section \ref{sec:invertiblemap}. In exploratory numerical experiments, we observed no significant decrease in inferential accuracy relative to the fully Bayesian approach, likely due to working with a small number of hyperparameters in $\bftheta$.

\section{Bayesian transport maps for large spatial fields\label{sec:spatial}}

Now assume that $\by = (y_1,\ldots,y_N)^\top$ consists of spatial observations or computer-model output at spatial locations $\bs_1,\ldots,\bs_N$ in a region or domain $\domain \subset \mathbb{R}^\text{dim}$.
We assume Bayesian transport maps as in Section \ref{sec:regression}, with regressions of the form \eqref{eq:regi} in $(i-1)$-dimensional space for $i=1,\ldots,N$. As $N$ is very large in many relevant applications, we will specify priors distributions of the form described in Section \ref{sec:priors} that induce substantial regularization and sparsity, as a function of hyperparameters $\bftheta = (\theta_{\sigma,1},\theta_{\sigma,2},\theta_{d,1},\theta_{d,2},\theta_\gamma,\theta_q)$ to be introduced in Sections \ref{sec:prior_s}--\ref{sec:prior_f}.

\subsection{Maximin ordering and nearest neighbors\label{sec:maximin}}

\begin{figure}
\centering
	\begin{subfigure}{.24\textwidth}
	\centering
 	\includegraphics[width =.99\linewidth]{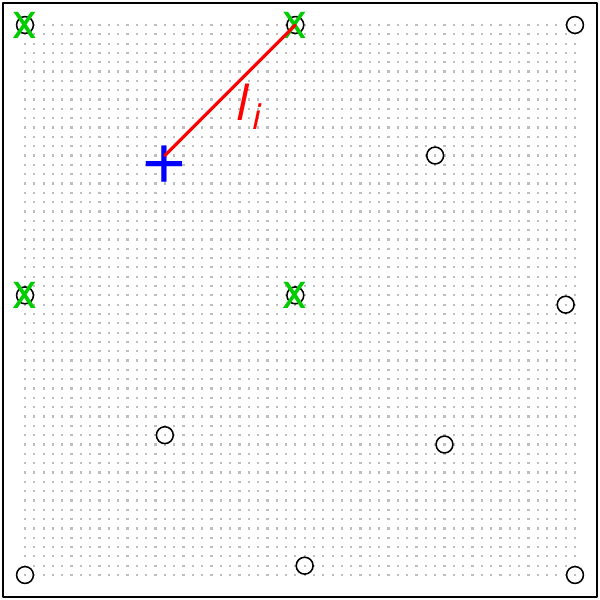}
	\caption{$i=13$}
	\label{fig:mm1}
	\end{subfigure}%
\hfill
\centering
	\begin{subfigure}{.24\textwidth}
	\centering
 	\includegraphics[width =.99\linewidth]{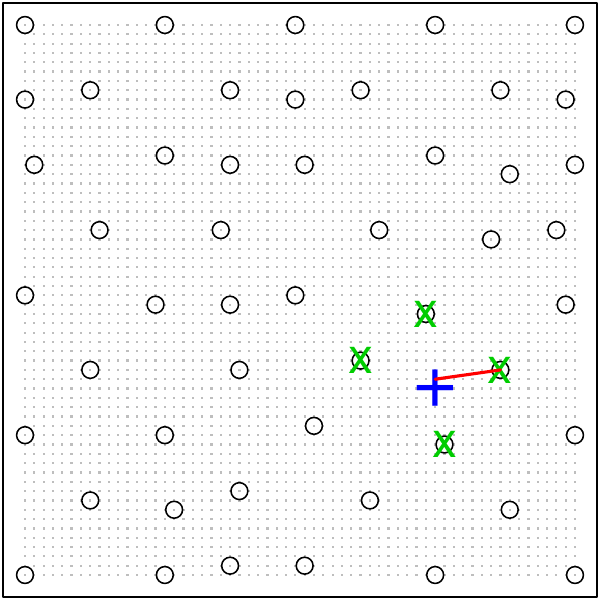}
	\caption{$i=51$}
	\label{fig:mm2}
	\end{subfigure}%
\hfill
\centering
	\begin{subfigure}{.24\textwidth}
	\centering
 	\includegraphics[width =.99\linewidth]{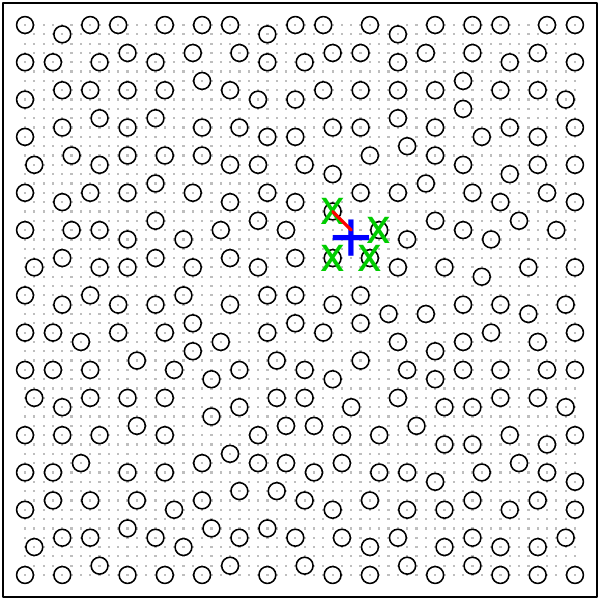}
	\caption{$i=290$}
	\label{fig:mm3}
	\end{subfigure}%
\hfill
	\begin{subfigure}{.24\textwidth}
	\centering
	\includegraphics[width =.99\linewidth]{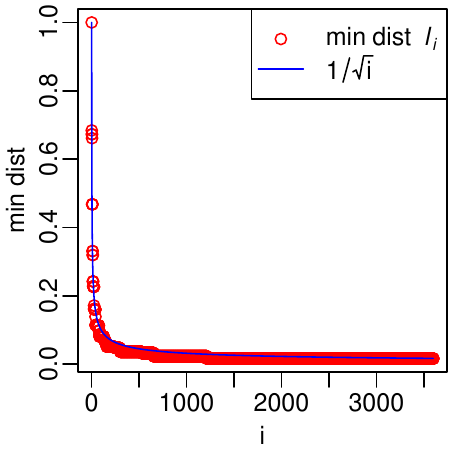}
	\caption{scale decay}
	\label{fig:scale}
	\end{subfigure}%
  \caption{Maximin ordering (Section \ref{sec:maximin}) for locations on a grid (small gray points) of size $N=60 \times 60 = 3{,}600$ on a unit square, $[0,1]^\text{dim}$ with $\text{dim}=2$. (a)--(c): The $i$th ordered location (${\color{blue}+}$), the previous $i-1$ locations (${\color{black}\circ}$), including the nearest $m=4$ neighbors (${\color{green}\mathbf{x}}$) and the distance $\ell_i$ to the nearest neighbor ({\color{red}---}). (d): For $i=1,\ldots,N$, the length scales (i.e., minimum distances) decay as $\ell_i = i^{-1/\text{dim}}$.}
\label{fig:maxmin}
\end{figure}

A triangular map $\map(\by)$ as in \eqref{eq:map} depends on the ordering of the variables $y_1,\ldots,y_N$.
We assume a maximum-minimum-distance (maximin) ordering of the corresponding locations $\bs_1,\ldots,\bs_N$ (see Figure \ref{fig:maxmin}), in which we sequentially choose each location to maximize the minimum distance to all previously ordered locations.
Specifically, the first index $i_1$ is chosen arbitrarily (e.g., $i_1=1$), and then the subsequent indices are selected as
$
i_j = \argmax_{i \, \notin \, \mathcal{I}_{j}} \,\, \min_{j \, \in \, \mathcal{I}_{j}} \|\bs_i - \bs_j\|
$
for $j = 2, \ldots , N$, where $\mathcal{I}_{j} = \{i_1 , \ldots , i_{j-1}\}$. For notational simplicity, we assume throughout that $\by=(y_1,\ldots,y_N)$ follows maximin ordering (i.e., $y_j = y_{i_j}$).
Define $c_i(k)$ as the index of the $k$th nearest (previously ordered) neighbor of the $i$th location (and so $\bs_{c_i(1)},\ldots,\bs_{c_i(4)}$ are indicated by ${\color{green}\mathbf{x}}$ in Figure \ref{fig:maxmin}). 

The maximin ordering can be interpreted as a multiresolution decomposition into coarse scales early in the ordering and fine scales later in the ordering. In particular, the minimal pairwise distance $\ell_i = \| \bs_i - \bs_{c_i(1)} \|$ among the first $i$ locations of the ordering decays roughly as $\ell_i \propto i^{-1/\text{dim}}$, where $\text{dim}$ here is the dimension of the spatial domain (see Figure \ref{fig:scale}).
As a result of the maximin ordering, the $i$th regression in \eqref{eq:regi} can be viewed as a spatial prediction at location $\bs_i$ based on data at locations $\bs_1,\ldots,\bs_{i-1}$ that lie roughly on a regular grid with distance (i.e., scale) $\ell_i$. 

When the variables $y_1,\ldots,y_N$ are not associated with spatial locations or when Euclidean distance between the locations is not meaningful (e.g., nonstationary, multivariate, spatio-temporal, or functional data), the maximin and neighbor ordering can be carried out based on other distance metrics, such as $(1 - |\text{correlation}|)^{1/2}$ based on some guess or estimate of the correlation between variables \citep{Kang2021,Kidd2020}.

\subsection{Priors on the conditional non-Gaussianity \texorpdfstring{$\sigma_i^2$}{}\label{sec:prior_s}}

In \eqref{eq:kernel}, $\sigma_i^2$ determines the degree of nonlinearity in $f_i$; hence, $\sigma_i^2,\ldots,\sigma_N^2$ together determine the conditional non-Gaussianity in the distribution of $\by_{i:N}$ given $\by_{1:i-1}$.
A priori, we assume that the degree of nonlinearity decays polynomially with length scale $\ell_i$, namely $\sigma_i^2 = e^{\theta_{\sigma,1}} \ell_i^{\theta_{\sigma,2}}$, which allows the conditional distributions of $\by_{i:N}$ given $\by_{1:i-1}$ to be increasingly Gaussian as $i$ increases, as a function of hyperparameters $\theta_{\sigma,1},\theta_{\sigma,2}$.

This prior assumption is motivated by the behavior of stochastic processes with quasiquadratic loglikelihoods.
A quasiquadratic loglikelihood of order $r$ is the sum of a quadratic leading-order term that depends on the $r$-th derivatives of the process, and a nonquadratic term that may only depend on derivatives up to order $r - 1$.
Gaussian smoothness priors (with quadratic loglikelihoods) such as the Mat{\'e}rn model \citep{Whittle1954,whittle1963stochastic} are closely related to linear elliptic PDEs. 
They can formally be thought of as having log-densities $-\langle u, Au \rangle / 2 - \langle u, b \rangle$ that are maximized by solutions of the linear equation $Au = b$.
Similarly, the maximizers of quasiquadratic log-densities $-\langle u, L(D^r u) / 2 \rangle - V(D^{r-1} u, \ldots u)$ are solutions of quasilinear PDEs $L(D^r u) = - \frac{d}{du} V(D^{r-1}u, \ldots u)$.
A wide range of physical phenomena is governed by quasilinear PDEs.
For instance, the Cahn-Hilliard \citep{cahn1958free} and Allen-Cahn \citep{allen1972ground} equations describe phase separation in multi-component systems, the Navier-Stokes equation describes the dynamics of fluids, and the F{\"o}ppl-von K{\'a}rm{\'a}n equations describe the large deformations of thin elastic plates.
If the order of a data-generating quasilinear PDE is known, a Mat{\'e}rn model with matching regularity is a sensible choice to model the leading-order behavior. 
However, most of the time, the observations will not arise from a known PDE model, and so the above arguments primarily motivate us to expect screening and power laws, without quantifying their effects.

In its simplest form, the mechanism relating local conditioning and approximate Gaussianity is captured by the classical Pointcar{\'e} inequality \citep{adams2003sobolev}.
\begin{lemma}[Poincar{\'e} inequality]
    Let $\Omega \subset \mathbb{R}^d$ be a Lipschitz-bounded domain with diameter $\ell$, let $u$ and its first derivative be square-integrable, and let $u_{\Omega} = \frac{1}{|\Omega|} \int_{\Omega} u(x) dx $ be the mean of $u$ over $\Omega$.
    Then, we have 
    \begin{equation}
        \|u - u_{\Omega}\|_{L^2\left(\Omega\right)} \leq \ell \|\nabla u\|_{L^2\left(\Omega\right)}.
    \end{equation}
\end{lemma} 
The Poincar{\'e} inequality directly implies the following corollary.
\begin{corollary}
Let $\Omega$ be a Lipschitz-bounded domain. Let $\tau$ be a partition of $\Omega$ into Lipschitz-bounded subdomains with diameter upper-bounded by $\ell$, and assume that $u,v$ and their first derivatives are square-integrable and satisfy $\int_t (u - v) dx = 0$ for all $t \in \tau$.
Then,
\begin{equation}
        \|u - v\|_{L^2\left(\Omega\right)} \leq \ell \|\nabla u - \nabla v\|_{L^2\left(\Omega\right)}. 
\end{equation}
\end{corollary}
This means that after conditioning a stochastic process on averages of diameter $\ell \ll 1$, even a minor perturbation in $u$ results in a large change of $\nabla u$.
For a quasiquadratic likelihood of order $r=1$ with lower-bounded curvature of the quadratic part, a perturbation that effects even a minor change in the nonlinear part, depending only on $u$, must effect a major change in the leading-order, quadratic term, assuming the curvature of the latter is bounded from below. 
Under suitable growth conditions on the nonlinear term, this means that the conditional density of a quasiquadratic likelihood of order $r = 1$ is dominated by the leading-order quadratic term as $\ell$ approaches zero. 
As a result, the conditional stochastic process is approximately Gaussian.
This is illustrated in a numerical example in Figure \ref{fig:CHconditioning}. Additional details are provided in Appendix \ref{app:quasilinear}.
Using generalizations of the Poincar{\'e} inequality to $r > 1$ and point-wise measurements \citep[e.g.,][Thm.~5.9(2)]{Schafer2017}, the above argument can be extended to the setting of $r >1$ and conditioning on point sets with distance $\ell$ (instead of local averages).

\begin{figure}
    \centering
    \includegraphics[width=0.24\textwidth]{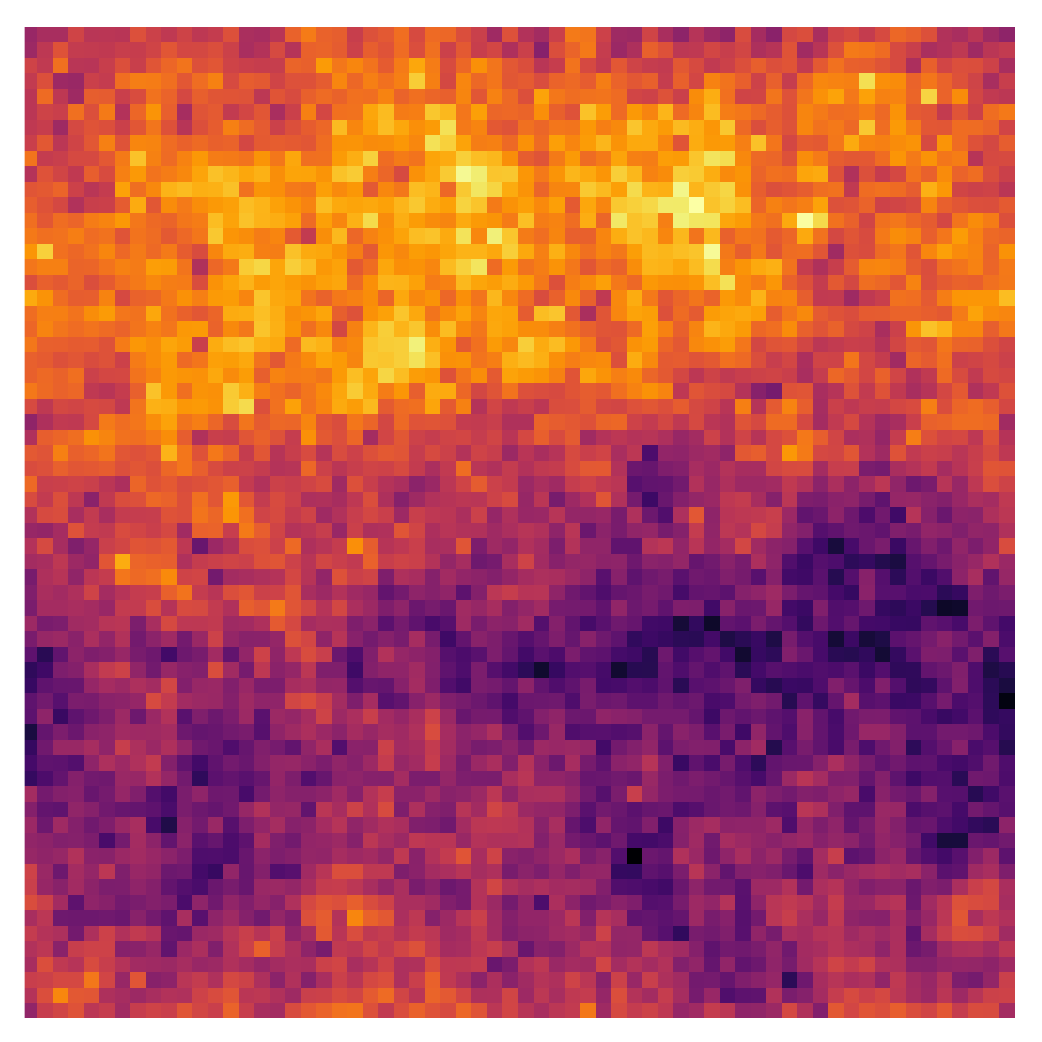}
    \includegraphics[width=0.24\textwidth]{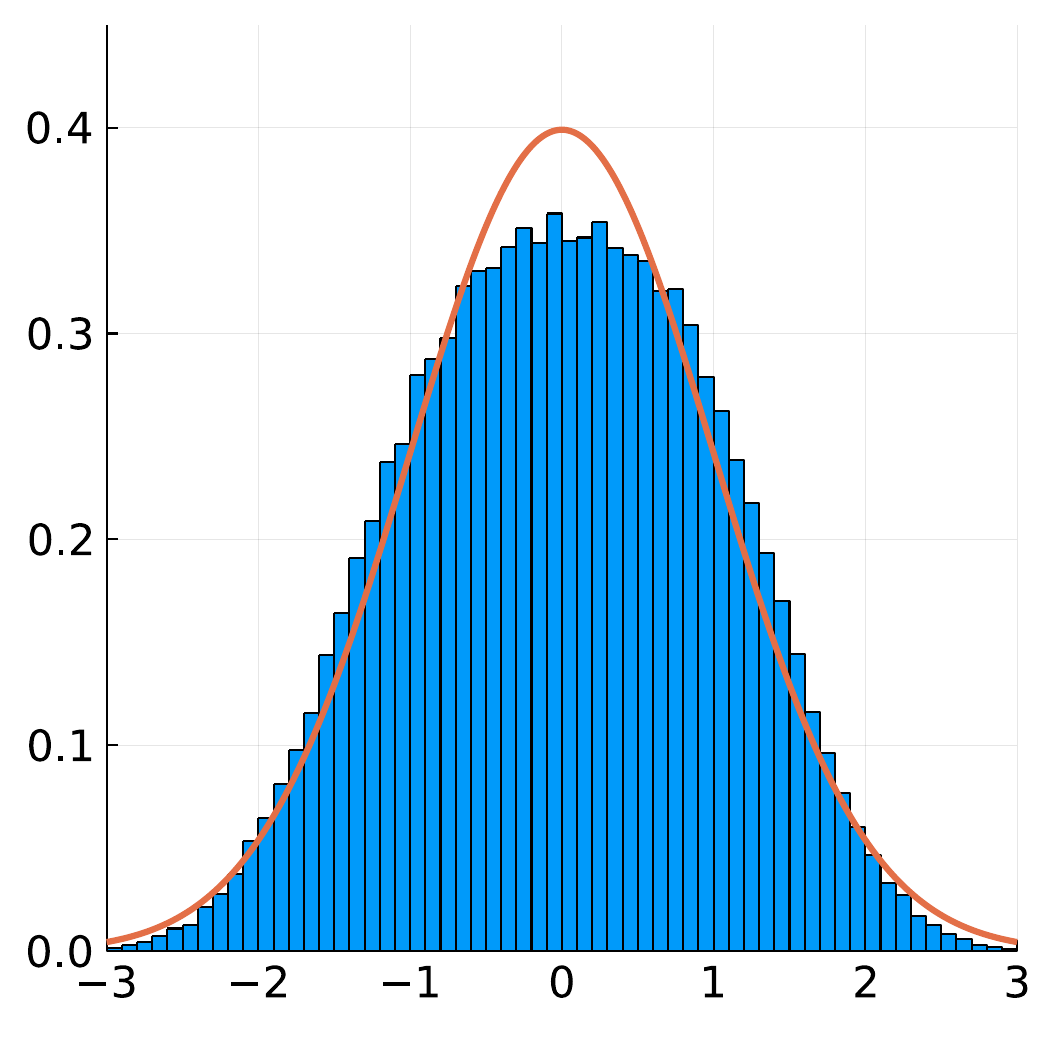}
    \includegraphics[width=0.24\textwidth]{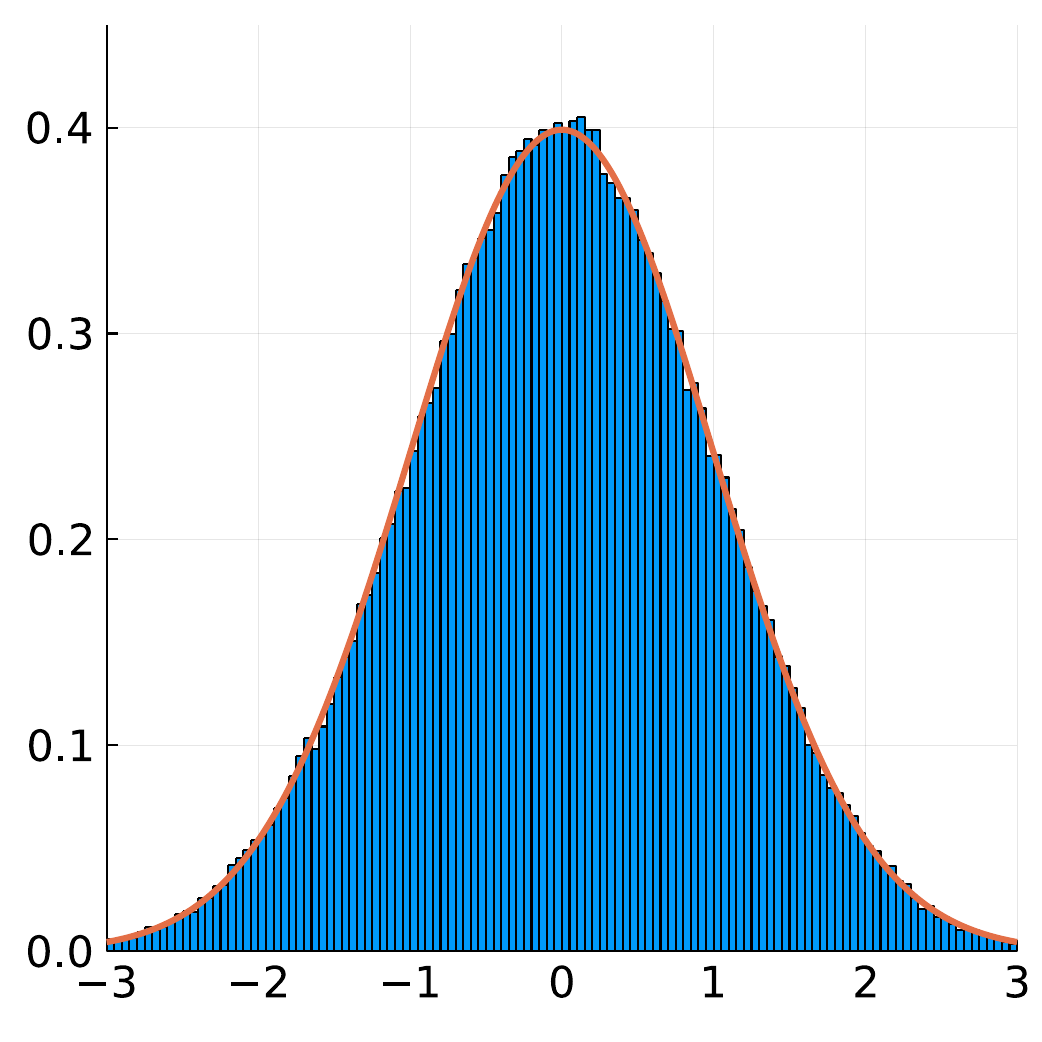}
    \includegraphics[width=0.24\textwidth]{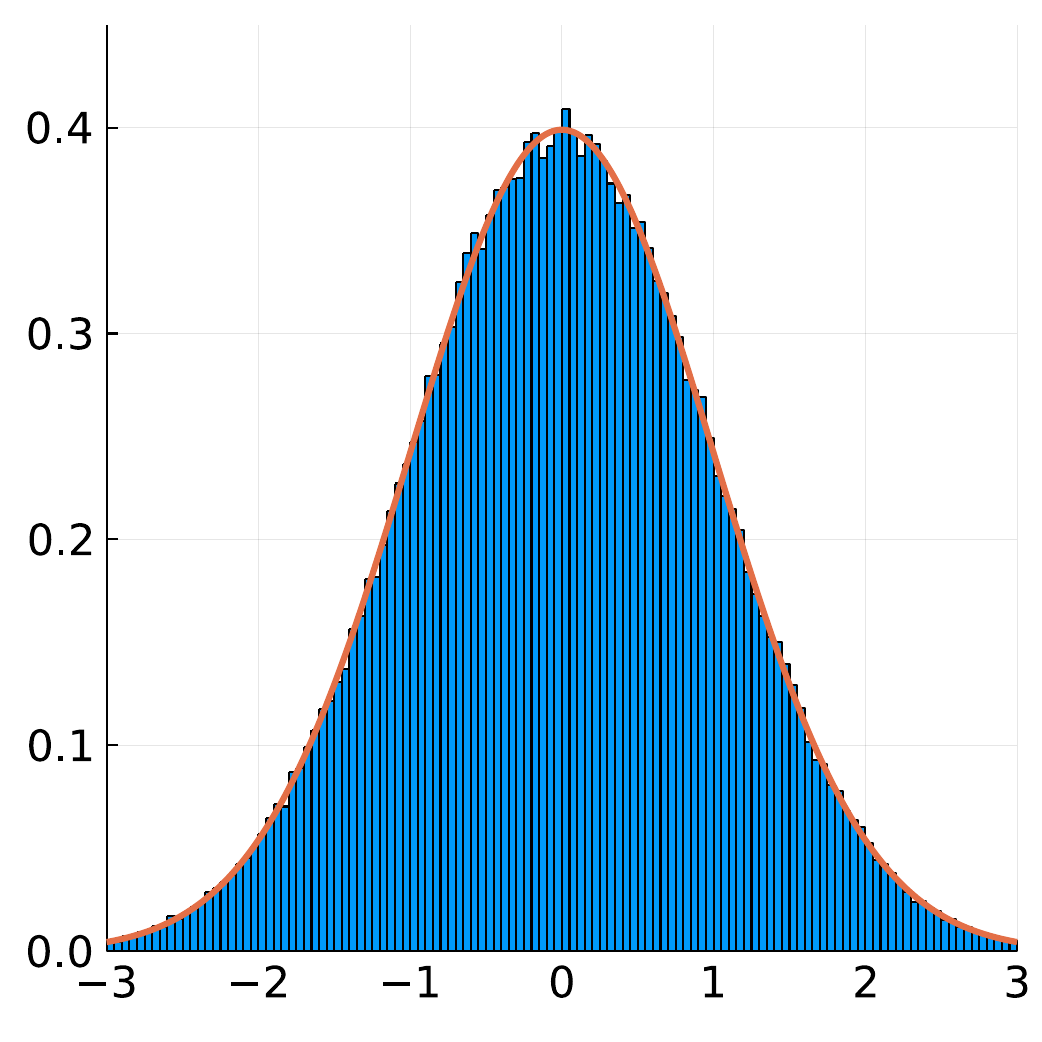}
    \caption{Samples from the non-Gaussian process in \eqref{eqn:CHdensity} feature regions of negative (dark) and positive (light) values (first panel). The distribution at a given location is a mixture of these two possibilities and thus non-Gaussian (second panel). By contrast, after conditioning on averages over regions of size $\ell = 2^{-1}$ (third panel) or $\ell = 2^{-5}$ (fourth panel), the conditional distribution is close to Gaussian, as these averages determine with high probability whether the location is in a positive or negative region. (See Appendix \ref{app:quasilinear} for details.)} 
    \label{fig:CHconditioning}
\end{figure}

\subsection{Priors on the conditional variances \texorpdfstring{$d_i^2$}{}\label{sec:prior_d}}

As we have argued in Section \ref{sec:prior_s}, even non-Gaussian stochastic processes with quasiquadratic loglikelihoods exhibit conditional near-Gaussianity on fine scales. Thus, we will now describe prior assumptions for the $d_i^2$ and $f_i$ in \eqref{eq:regi} that are motivated by the behavior of a transport map $\map$ for a Gaussian target distribution with Mat{\'e}rn covariance (see Figure \ref{fig:matern}), which is a highly popular assumption in spatial statistics. 
The Mat{\'e}rn covariance function is also the Green's function of an elliptic PDE \citep[][]{Whittle1954,whittle1963stochastic}. 

\citet[][Thm.~2.3]{Schafer2017} show that Gaussian processes with covariance functions given by the Green's function of elliptic PDEs of order $r$ have conditional variance of order $\ell_i^{2r}$ when conditioned on the first $i$ elements of the maximin ordering (see Figure \ref{fig:matern1}). 

Hence, for the noise or conditional variances $d_i^2 \sim \mathcal{IG}(\alpha_i,\beta_i)$ as in Section \ref{sec:priors}, we set $E(d_i^2) = \beta_i/(\alpha_i -1) = e^{\theta_{d,1}} \ell_i^{\theta_{d,2}}$. Assuming the prior standard deviation of $d_i^2$ to be equal to $g$ times the mean, we obtain $\alpha_i = 2+1/g^2$ and $\beta_i = e^{\theta_{d,1}} \ell_i^{\theta_{d,2}} (1+1/g^2)$. For our numerical experiments, we chose $g=4$ to obtain a relatively vague prior for the $d_i^2$.

\begin{figure}
\centering
	\begin{subfigure}{.37\textwidth}
	\centering
  	\includegraphics[trim=4mm 13mm 9mm 22mm, clip,width =.99\linewidth]{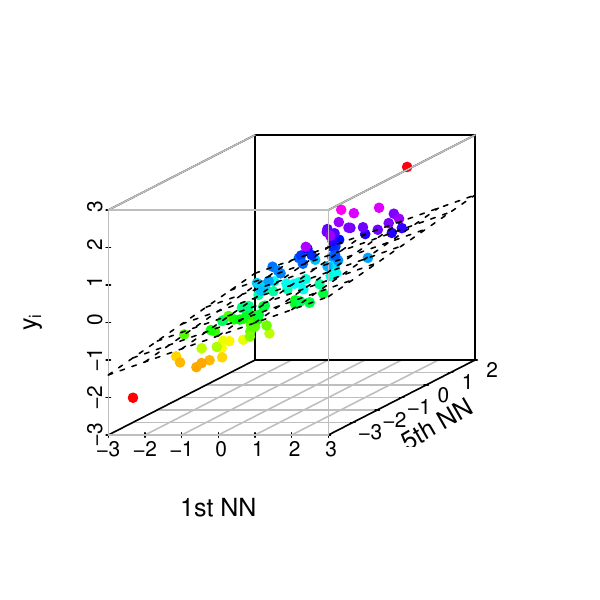} 
	\caption{$\by_i$ vs $\by_{c_i(1)}$, $\by_{c_i(5)}$ for $i=290$}
	\label{fig:maternscatter}
	\end{subfigure}%
\hfill
	\begin{subfigure}{.27\textwidth}
	\centering
 	\includegraphics[width =.99\linewidth]{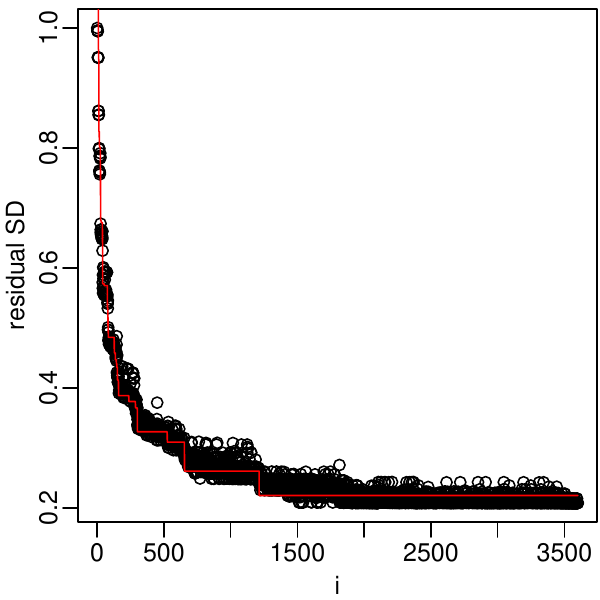}
	\caption{$d_{i}$ ($\circ$) \, and \,  $e^{\theta_{d,1}} \ell_i^{\theta_{d,2}}$ ({\color{red}---})}
	\label{fig:matern1}
	\end{subfigure}%
\hfill
	\begin{subfigure}{.27\textwidth}
	\centering
 	\includegraphics[width =.99\linewidth]{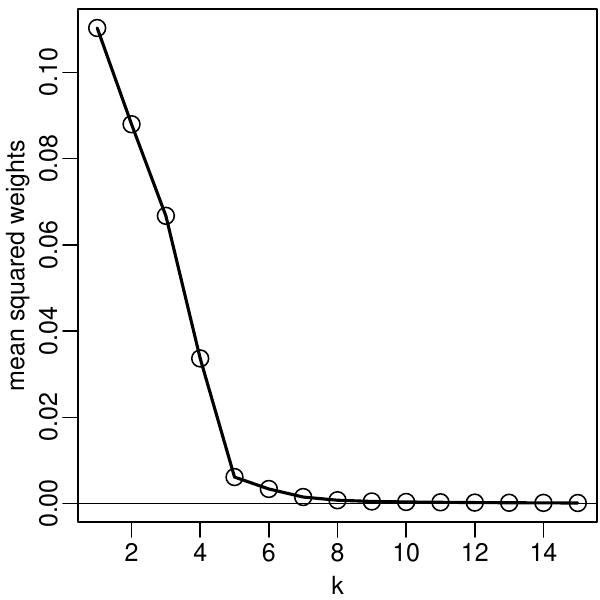}
	\caption{Avg.\ $\{b_{i,k}^2: i\!=\!1,\ldots,N\}$}
	\label{fig:matern2}
	\end{subfigure}%
  \caption{For a Gaussian process with exponential covariance on the grid and with the ordering from Figure \ref{fig:maxmin}, expressing the joint distribution $p(\by)$ using a transport map as in \eqref{eq:map}--\eqref{eq:condmap} results in a series of regressions as in \eqref{eq:regi} with linear predictors, $f_i(\by_{1:i-1}) = \sum_{k=1}^{i-1} y_{c_i(k)}b_{i,k}$, where $c_i(k)$ indicates the $k$th nearest (previously ordered) neighbor of the $i$th location.
  (For non-Gaussian $p(\by)$, the functions $f_i$ are nonlinear.)
  (a): For $n=100$ simulations, the values of $\by_i$ and its 1st and 5th nearest neighbor (NN) lie on a low-dimensional manifold; the regression plane (assuming all other variables to be fixed) indicates a stronger influence of the 1st NN (see the slope of the intersection of the regression plane with the front of the box) than of the 5th NN. 
  (b): The conditional standard deviations decay as a function of the length scale $\ell_i$ (see Figure \ref{fig:scale}). 
  (c) The squared regression coefficients decay rapidly as a function of neighbor number $k$.
}
\label{fig:matern}
\end{figure}

\subsection{Priors on the regression functions \texorpdfstring{$f_i$}{}\label{sec:prior_f}}

The regression functions $f_i: \mathbb{R}^{i-1} \rightarrow \mathbb{R}$ in \eqref{eq:regi} were specified to be GPs in $(i-1)$-dimensional space in Section \ref{sec:priors}. For the covariance function in \eqref{eq:kernel}, we assume that $\rho_i(\by_{1:i-1},\by_{1:i-1}') = \rho\big(h_i(\by_{1:i-1},\by_{1:i-1}')/\gamma\big)$, where $h_i^2(\by_{1:i-1},\by_{1:i-1}') = (\by_{1:i-1}-\by_{1:i-1}')^\top\bQ_i(\by_{1:i-1}-\by_{1:i-1}')$, $\gamma = \exp(\theta_\gamma)$ is a range parameter, and $\rho$ is an isotropic correlation function, taken to be Mat\'ern with smoothness 1.5 for our numerical experiments.

To make this potentially high-dimensional regression feasible, we again use the example of a spatial GP with Mat\'ern covariance to motivate regularization and sparsity via the relevance matrix $\bQ_i = \diag(q_{i,1}^2,\ldots,q_{i,i-1}^2)$.
We assume that the relevance of the $k$th neighbor (see Section \ref{sec:maximin}) decays exponentially as a function of $k$, such that $q_{i,c_i(k)}$ decays as $\exp(\theta_q k)$. This type of behavior, often referred to as the screening effect \citep[e.g.,][]{stein20112010}, is illustrated in Figure \ref{fig:matern2}, and it has been exploited for covariance estimation of a Gaussian spatial field by \citet{Kidd2020}. 
Recently, \citet{Schafer2017} proved exponential rates of screening for Gaussian processes derived from elliptic boundary-value problems; following the discussion in Section~\ref{sec:prior_s}, we expect similar conditional-independence phenomena to hold on the fine scales of processes with quasiquadratic loglikelihoods. As shown in Figure \ref{fig:precweights}, we also observed this behavior for climate data.

Given this exponential decay as a function of the neighbor number $k$, the relevance will be essentially zero for sufficiently large $k$, and so we achieve sparsity by setting 
\begin{equation}
q_{i,c_i(k)} = \begin{cases} \exp(\theta_q k), & k \leq m,\\ 0, & k>m, \end{cases}
\label{eq:sparsity}    
\end{equation}
where the sparsity parameter $m=\max \{k: \exp(\theta_q k) \geq \varepsilon\}$ is determined by the data through the hyperparameter $\theta_q$. We used $\varepsilon = 0.01$ for our numerical examples, which produced highly accurate inference and usually resulted in $m<10$.
Assumption \eqref{eq:sparsity} induces a sparse transport map, in that $f_i$ (and thus $\map_i$) depend on $\by_{1:i-1}$ only through the $m$ nearest neighbors $y_{c_i(1)},\ldots,y_{c_i(m)}$, where $\rho_i$ is isotropic as a function of the scaled inputs $y_{c_i(k)}/q_{i,c_i(k)}$.
Sparsity in the transport map is equivalent to an assumption of ordered conditional independence.
Similar ordered-conditional-independence assumptions are also popular for Vecchia approximations of Gaussian fields with parametric covariance functions.

Identifying the regression functions $f_i$ in $m$-dimensional space is further aided by the data approximately concentrating on a lower-dimensional manifold due to the strong dependence between most $y_{c_i(k)}$ and $y_{c_i(l)}$ for small $k,l \leq m$ (e.g., see Figure \ref{fig:maternscatter}).

\subsection{Inference\label{sec:spatialinference}}

\begin{algorithm}[t]
\caption{Inference for the spatial transport map}
\begin{algorithmic}[1]
\STATE Order $y_1,\ldots,y_N$ in maximin ordering and compute scales $\ell_i$ and nearest-neighbor indices $c_i(1),\ldots,c_i(m_{\text{max}})$ (e.g., $m_{\text{max}}=30$) for each $i=1,\ldots,N$ (see Section \ref{sec:maximin})
\STATE Compute $\hat\bftheta = \argmax_{\bftheta} \log p(\bY)$ via stochastic gradient ascent, where $p(\bY) \textstyle \propto \prod_{i=1}^N \big( \, |\bG_i|^{-1/2} \times ({\beta_i^{\alpha_i}}/{\tilde\beta_i^{\tilde\alpha_i}}) \times {\Gamma(\tilde\alpha_i)}/{\Gamma(\alpha_i)} \, \big)$, 
with 
$\bftheta = (\theta_{\sigma,1},\theta_{\sigma,2},\theta_{d,1},\theta_{d,2},\theta_\gamma,\theta_q)$,
$\tilde\alpha_i = \alpha_i + n/2$, 
$\tilde\beta_i = \beta_i + \by_i{}^\top \bG_i^{-1} \by_i/2$,
$\alpha_i = 2+1/g^2$,
$\beta_i = e^{\theta_{d,1}} \ell_i^{\theta_{d,2}} (1+1/g^2)$,
$g=4$,
$\bG_i = (C_i(\by^{(j)},\by^{(l)}) )_{j,l=1,\ldots,n}/(e^{\theta_{d,1}} \ell_i^{\theta_{d,2}}) + \bI_n$, 
$C_i(\by^{(j)},\by^{(l)}) = \sum_{k=1}^m \tilde{y}_{c_i(k)}^{(j)}\tilde{y}_{c_i(k)}^{(l)} + \allowbreak \sigma^2_i \, \rho\big((\sum_{k=1}^m (\tilde{y}_{c_i(k)}^{(j)}-\tilde{y}_{c_i(k)}^{(l)})^2)^{1/2}/\gamma\big)$, 
$\tilde{y}_{c_i(k)}^{(j)} = y_{c_i(k)}^{(j)} e^{\theta_q k}$, 
$m=\max \{k: e^{\theta_q k} \geq 0.01\}$,
$\sigma^2_i = e^{\theta_{\sigma,1}} \ell_i^{\theta_{\sigma,2}}$, 
$\gamma = e^{\theta_\gamma}$, $\rho(x)=(1+x\sqrt{3})\exp(-x\sqrt{3})$
\STATE Use fitted map as desired. For example, generate a new sample $\by^\star=\pmap_{\hat{\bftheta}}^{-1}(\bz^\star)$ using \eqref{eq:invmap} based on $\bz^\star \sim \normal_N(\bfzero,\bI_N)$.
\end{algorithmic}
\label{alg:inf}
\end{algorithm}

Based on the prior distributions in Sections \ref{sec:prior_s}--\ref{sec:prior_f}, we can carry out inference and compute the transport map as in Section \ref{sec:invertiblemap}. 
The prior distributions depend on a vector of hyperparameters, $\bftheta = (\theta_{\sigma,1},\theta_{\sigma,2},\theta_{d,1},\theta_{d,2},\theta_\gamma,\theta_q)$. When making inference on $\bftheta$ as described in Section \ref{sec:hyper}, we effectively let the training data $\bY$ decide the degree of sparsity (through $\theta_q$ via $m$) and the degree of nonlinearity (through $\theta_{\sigma,1},\theta_{\sigma,2}$ via $\sigma_i$).
Algorithm \ref{alg:inf} summarizes the inference procedure.
Figure \ref{fig:sine} illustrates estimation of transport-map components in a simulated example.

\begin{figure}
\centering
	\begin{subfigure}{.22\textwidth}
	\centering
 	\includegraphics[width =.99\linewidth]{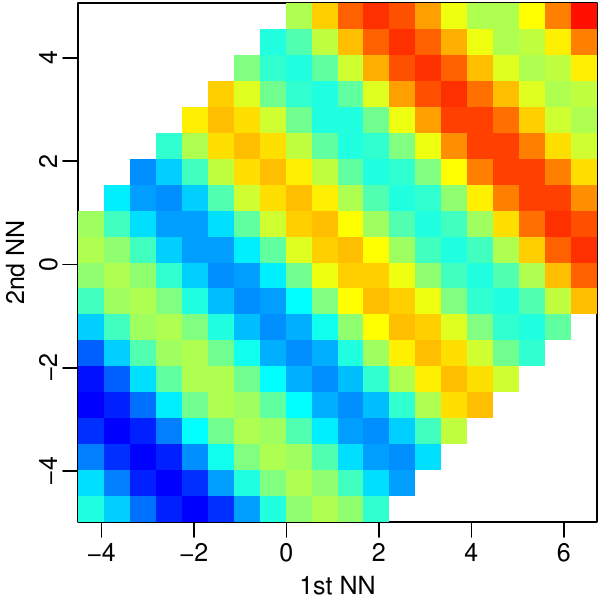}
	\caption{True}
	\label{fig:sinetrue}
	\end{subfigure}%
\hfill	
	\begin{subfigure}{.22\textwidth}
	\centering
 	\includegraphics[width =.99\linewidth]{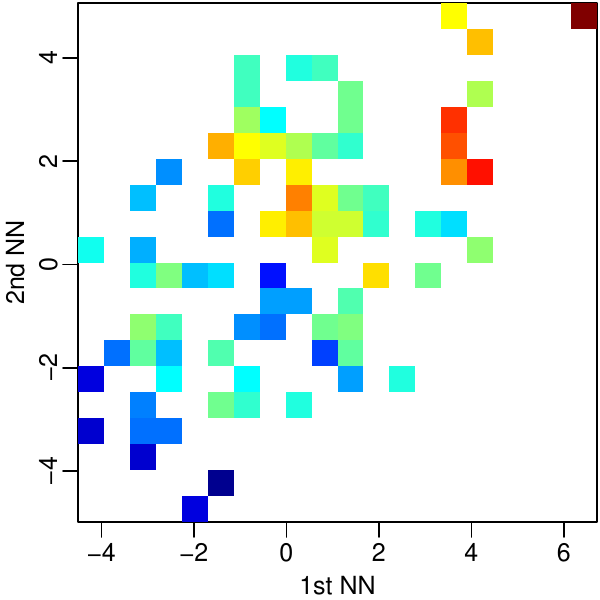}
	\caption{Data}
	\label{fig:sinedata}
	\end{subfigure}%
\hfill	
	\begin{subfigure}{.22\textwidth}
	\centering
 	\includegraphics[width =.99\linewidth]{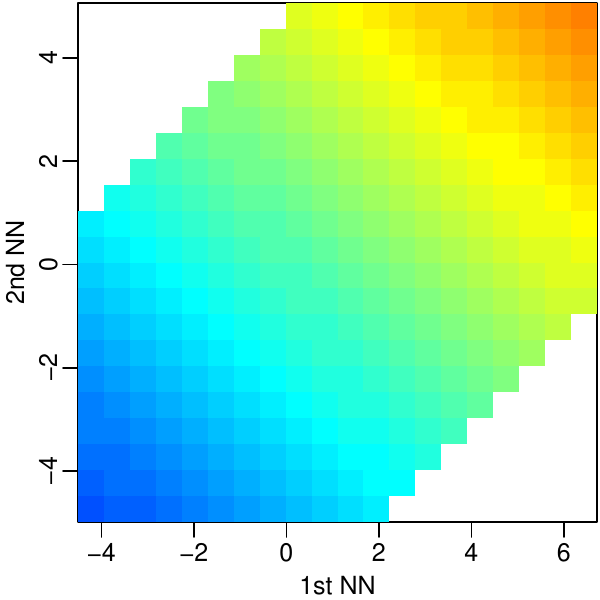}
	\caption{Linear}
	\label{fig:sinelin}
	\end{subfigure}%
\hfill	
	\begin{subfigure}{.22\textwidth}
	\centering
 	\includegraphics[width =.99\linewidth]{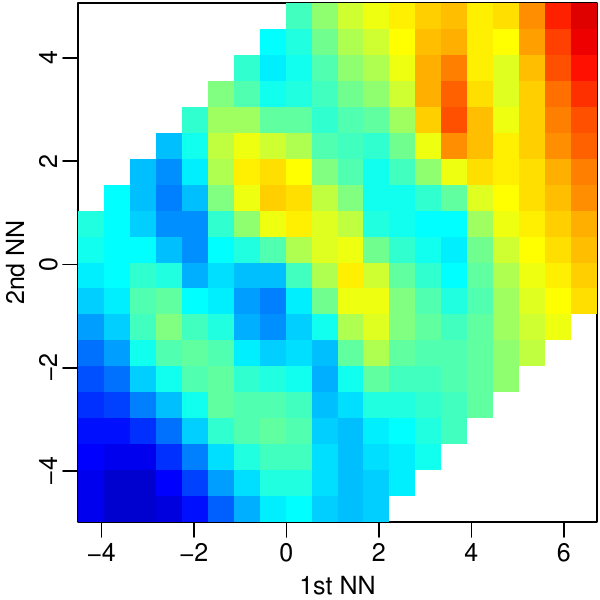}
	\caption{Nonlinear}
	\label{fig:sinenonlin}
	\end{subfigure}%
\hfill
\begin{minipage}{.042\textwidth}
~

\vspace{-10mm}

\includegraphics[trim=78mm 0mm 5mm 0mm, clip,width =.93\linewidth]{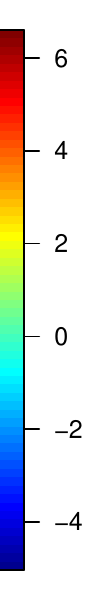} 
\end{minipage}	
  \caption{Simulation from a nonlinear map with sine structure in $f_i$, described as NR900 in Section \ref{sec:simstudy}. For $n=100$ and $i=80$, $y_i$ versus its 1st and 2nd nearest neighbor (NN): true $f_i$ (a), observations $\by_i$ (b), together with linear (c) and nonlinear (d) fit (i.e., posterior means) of $f_i$, with further variables in $\by_{1:i-1}$ held at their mean levels. The linear map in (c) is estimated under the restriction $\sigma_i=0$. In (d), we have a nonlinear regression in $79$-dimensional space, with $m=5$ active variables in the estimated (via $\bftheta$) nonlinear model.}
\label{fig:sine}
\end{figure}

Due to the sparsity assumption in \eqref{eq:sparsity}, the computational complexity is lower than in Section \ref{sec:invertiblemap}; specifically, determining $\pmap_i$ now only requires $\order(n^3 + m n^2)$ time, again in parallel for each $i=1,\ldots,N$. Each application of the transport map or its inverse then requires $\order(N(n^2 + mn))$ time.
The maximin ordering and nearest neighbors can also be computed in quasilinear time in $N$ \citep[][Alg.~7]{Schafer2020}.

In Section \ref{sec:invertiblemap}, we discussed using $\pmap$ in \eqref{eq:singlemap} to transform the non-Gaussian $\by^\star$ to standard Gaussian map coefficients $\bz^\star = \pmap(\by^\star)$. This concept, which is illustrated in Figure \ref{fig:trans}, is especially interesting in our spatial setting. Due to the maximin ordering (Figure \ref{fig:maxmin}), the scales $\ell_i$ are arranged in decreasing order, and in our prior the $d_i^2$ also follow a decreasing stochastic order with $E(d_i^2) = e^{\theta_{d,1}} \ell_i^{\theta_{d,2}}$ (see, e.g., Figure \ref{fig:matern1}). Thus, we can view the map components as a form of nonlinear principal components (NPCs), with the map coefficients as the corresponding component scores.
For Gaussian processes with covariance functions given by the Green's function of elliptic PDEs, similar to the Mat{\'e}rn family, it can be shown that these principal components based on the maximin ordering are approximately optimal \citep{Schafer2017}.
For example, as illustrated in Appendix \ref{app:recon}, these NPCs can be used for dimension reduction by only storing or modeling the first $k$, say, map coefficients $\bz^\star_{1:k} = (z^\star_1,\ldots,z^\star_k)^\top$. 
Note that if we set $\bz_{k+1:N}^\star = \bfzero$, we assume $y_i^\star = \hat f_i(\by_{1:i-1}^\star)$ for $i>k$, which overestimates dependence and underestimates variability; hence, it is preferable to draw $\bz_{k+1:N}^* \sim \normal(\bfzero,\bI)$.
In addition to reducing storage, we can also use this approach for conditional simulation \citep[][Lemma 1]{Marzouk2016}, in which we fix the large-scale features of an observed field by fixing the first $k$ map coefficients (see Figure \ref{fig:precipcond} for an illustration).
To model a time series of spatial fields, we could assume a linear vector autoregressive model for the NPCs, such that the map coefficients at time $t+1$, say $\bz_{1:k}^{(t+1)}$, linearly depend on $\bz_{1:k}^{(t)}$.
When it is of interest to regress some response on a spatial field, one could also use the first $k$ map coefficients of the field as the covariates, similar to the use of function principal component scores in regression.

\begin{figure}
\centering
	\begin{subfigure}{.19\textwidth}
	\centering
 	\includegraphics[width =.99\linewidth]{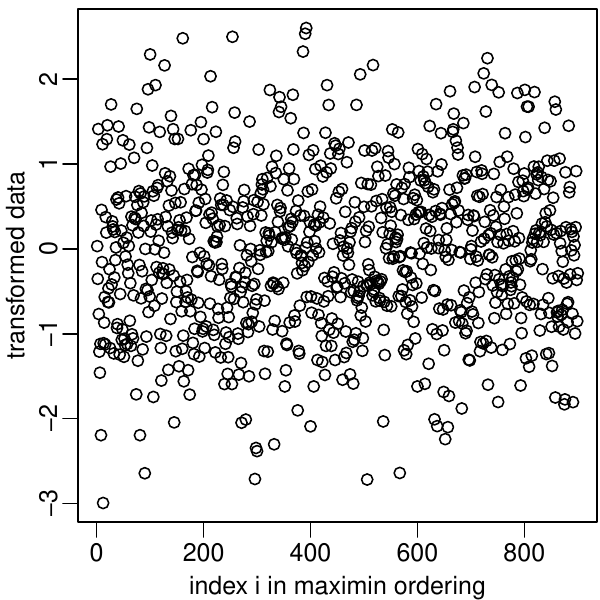}
	\caption{Transformed}
	\label{fig:transdata}
	\end{subfigure}%
\hfill	
	\begin{subfigure}{.19\textwidth}
	\centering
 	\includegraphics[width =.99\linewidth]{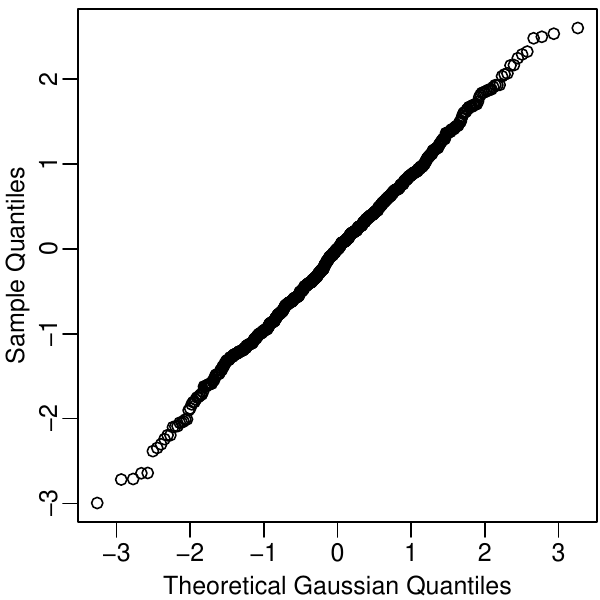}
	\caption{QQ plot}
	\label{fig:qq}
	\end{subfigure}%
\hfill	
	\begin{subfigure}{.19\textwidth}
	\centering
 	\includegraphics[width =.99\linewidth]{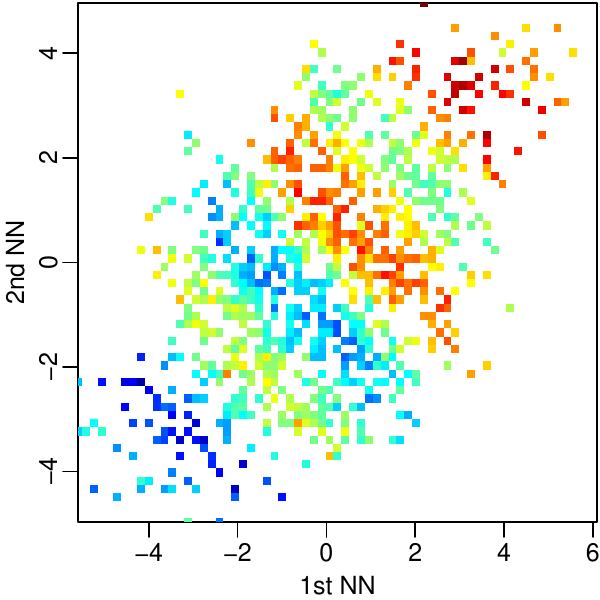}
	\caption{Test data}
	\label{fig:avgtest}
	\end{subfigure}%
\hfill	
	\begin{subfigure}{.19\textwidth}
	\centering
 	\includegraphics[width =.99\linewidth]{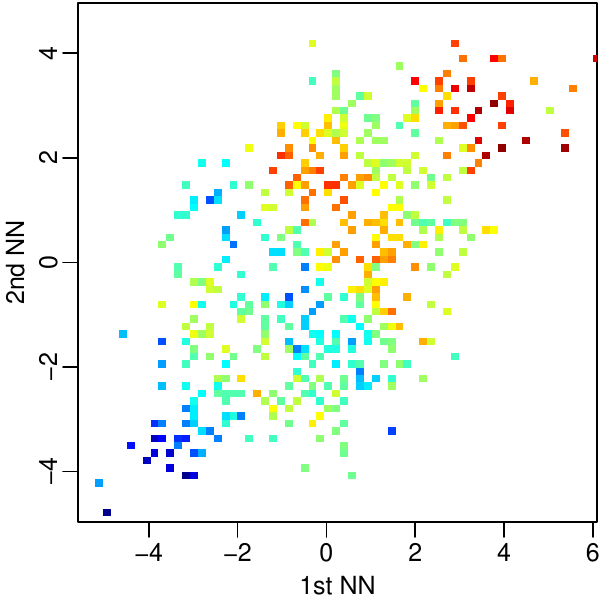}
	\caption{Reference avg}
	\label{fig:avglatent}
	\end{subfigure}%
\hfill	
	\begin{subfigure}{.19\textwidth}
	\centering
 	\includegraphics[width =.99\linewidth]{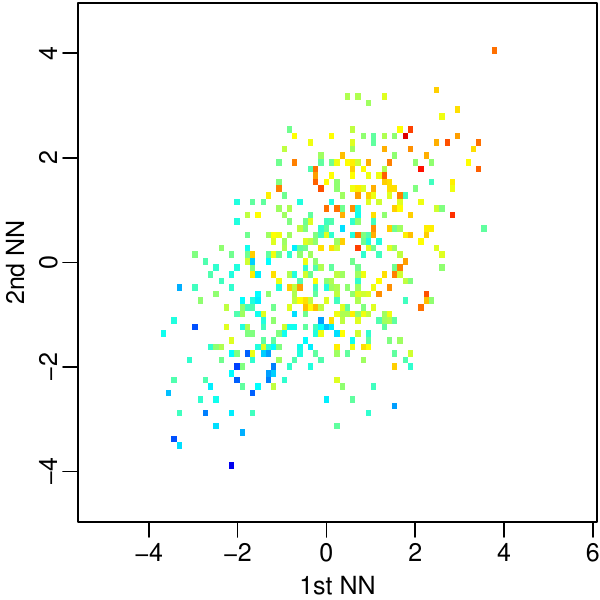}
	\caption{Original avg}
	\label{fig:avgorig}
	\end{subfigure}%
\hfill
\begin{minipage}{.036\textwidth}
~

\vspace{-9mm}

\includegraphics[trim=78mm 0mm 5mm 0mm, clip,width =.93\linewidth]{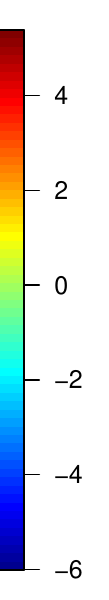} 
\end{minipage}	
  \caption{Illustration of map coefficients $\bz^\star = \pmap(\by^\star)$ (see Sections \ref{sec:invertiblemap} and \ref{sec:spatialinference}) for the simulated NR900 data using $\pmap$ inferred from $n=100$ training data. 
  (a)--(b): The $N=900$ map coefficients corresponding to one test sample are roughly i.i.d.\ Gaussian. 
  (c): For $1{,}000$ test samples, $y^*_{80}$ versus 1st and 2nd NNs (cf.~Figure \ref{fig:sine}).
  (d) When averaging pairs of two map-coefficient vectors in reference space and transforming back to the original space using \eqref{eq:invmap}, the sinusoidal relationship between $y^*_{80}$ and its NNs is preserved in the resulting 500 averages. 
  (e) When averaging test samples directly in original space, the nonlinear structure is lost.}
\label{fig:trans}
\end{figure}

\section{Non-Gaussian errors\label{sec:dpm}}

So far, we have focused on nonlinear, non-Gaussian dependence structures. The model described in Sections \ref{sec:bayesian} and \ref{sec:spatial} assumes Gaussian errors in the regressions \eqref{eq:regi}, which implies a marginal Gaussian distribution for $y_1$, the first variable in the maximin ordering. If this does not hold at least approximately, extensions based on additional marginal (i.e., pointwise) transformations, especially of the first few variables in the ordering, are straightforward. For example, assume that the model from Sections \ref{sec:bayesian} and \ref{sec:spatial} holds for $\by$, but that we actually observe $\tilde\by = \mathcal{G}(\by)$ such that $\tilde y_i = g_i(y_i)$. If the $g_i$ are one-to-one differentiable functions, the resulting posterior map is a simple extension of that in Proposition \ref{prop:maps}. The $g_i$ can be pre-determined (see Section \ref{sec:application} for an example with a log transform) or may depend on $\bftheta$ and thus be inferred based on a minor modification of the integrated likelihood in Proposition \ref{prop:lik}.

To increase flexibility of the marginal distributions, the GP errors $\epsilon_i^{(j)}$ can be modeled using Bayesian nonparametrics for all $i=1,\ldots,N$. More precisely, we will use Dirichlet process mixtures (DPMs). In \eqref{eq:regi}, we now assume that $f_i(\cdot) \sim \GP(0,C_i)$, and the $\epsilon_i^{(j)}$ are distributed according to a DPM for $j=1,\ldots,n$:
\begin{equation}
    \epsilon_i^{(j)} | \mu_i^{(j)},d_i^{(j)} \sim \normal(\mu_i^{(j)},(d_i^{(j)})^2), \quad (\mu_i^{(j)},(d_i^{(j)})^2) | \mathcal{F}_i \sim \mathcal{F}_i, \quad \mathcal{F}_i \sim \mathcal{DP}(\nig(\xi_i,\eta_i,\alpha_i,\beta_i),\zeta_i),
\end{equation}
where $\zeta_i$ is the concentration parameter, and the base measure $\nig(\xi_i,\eta_i,\alpha_i,\beta_i)$ is a normal-inverse-Gamma distribution with density
$
p(x,y) = \eta_i^{1/2} (2\pi y)^{-1/2} \beta_i^{\alpha_i}/\Gamma(\alpha_i) y^{-\alpha_i-1}\exp(-(2\beta_i+\eta_i(x-\xi_i)^2)/(2y)),
$
where we assume $\xi_i=0$.
The degree of non-Gaussianity allowed for the $\epsilon_i^{(j)}$ is determined by $\eta_i$ and $\zeta_i$. A small value of $\zeta_i$ concentrates the Dirichlet process near the NIG base measure, for which a large value of $\eta_i$ shrinks the $\mu_i^{(j)}$ toward zero. Thus, in the limit as $\zeta_i \rightarrow 0$ and $\eta_i \rightarrow \infty$, we obtain a model similar to that in Section \ref{sec:priors}
(except that here the $d_i^{(j)}$ do not appear in the variance of the GP $f_i$).
Conversely, for large $\zeta_i$ (or large $n$), the posterior of $\epsilon_i^{(j)}$ will be a Gaussian mixture that may differ substantially from the posterior implied by the model in Section \ref{sec:priors}.

For the spatial setting with maximin ordering of Section \ref{sec:spatial}, we can again find a sparse parameterization in terms of hyperparameters $\bftheta = (\theta_{\sigma,1},\theta_{\sigma,2},\theta_{d,1},\theta_{d,2},\theta_\gamma,\theta_q,\theta_{\zeta,1},\theta_{\zeta,2},\theta_{\eta,1},\theta_{\eta,2})$. We parameterize the $\alpha_i$, $\beta_i$, $C_i$ in terms of the first six hyperparameters as in Sections \ref{sec:prior_s}--\ref{sec:prior_f}. For the concentration parameter $\zeta_i = e^{\theta_{\zeta,1}}\ell_i^{\theta_{\zeta,2}}$, we allow increasing shrinkage toward Gaussianity for increasing $i$. We similarly set $\eta_i = e^{\theta_{\eta,1}}\ell_i^{\theta_{\eta,2}}$.
For this DPM model, we take a fully Bayesian perspective and assume an improper uniform prior for $\bftheta$ over $\mathbb{R}^{10}$.

The resulting model is fully nonparametric with the exception of the additivity assumption in \eqref{eq:condmap}. Specifically, due to the nonparametric nature of the DPM, the universal approximation property of GPs \citep{Micchelli2006}, and nonzero prior probability for the dense (non-sparse) transport map, the posterior distribution obtained using this model contracts (for $n\rightarrow \infty$ and fixed $N$) to the Kullback-Leibler (KL) projection of the actual distribution of $\by$ onto the space of distributions that can be described by a transport map whose components are additive in the $i$th argument as in \eqref{eq:condmap}, due to the KL optimality of the Knothe-Rosenblatt map \citep[][Sec.~4.1]{Marzouk2016}. In other words, as the number of replicates increases, the learned distribution gets as close as possible to the truth under the additivity restriction.

Inference for our DPM model cannot be carried out in closed form anymore and instead relies on a Metropolis-within-Gibbs Markov chain Monte Carlo (MCMC) sampler. We can also compute and draw samples from the posterior predictive distribution
\[
\textstyle
p(\by^\star|\bY) = \prod_{i=1}^N p(y_i^\star|\by_{1:i-1}^\star,\bY),
\]
for which each $p(y_i^\star|\by_{1:i-1}^\star,\bY)$ is approximated as a Gaussian mixture based on the MCMC output.
Details for the MCMC procedure and the posterior predictive distribution are given in Appendix \ref{app:dpmgibbs}.

In the spatial setting with sparsity parameter $m$, each MCMC iteration still has time complexity $\order(N(n^3 + n^2m))$ and the computations within each iteration are highly parallel; however, the actual computational cost for this sampler is much higher (typically, roughly two orders of magnitude higher) than for the empirical Bayes approach in Section \ref{sec:hyper} due to the large number of MCMC iterations required.
Because of this larger computational expense and the loss of a closed-form transport map for the DPM model, we recommend the empirical Bayes approach (potentially after a pre-transformation $\mathcal{G}$ as described above) as the first option in most large-scale applications; the DPM model is most useful for settings in which its computational expense is not crucial, the training size $n$ is sufficiently large to discern non-Gaussian error structure, and only posterior sampling (as opposed to other functions that transport maps can provide) is of interest.

\section{Simulation study \label{sec:simstudy}}

We compared the following methods:
\begin{description}[itemsep=1pt,topsep=2pt,parsep=1pt]
\item[\texttt{nonlin}:] Our method with Bayesian uncertainty quantification described in Section \ref{sec:spatial}.
\item[\texttt{S-nonlin}:] Simplified version of \texttt{nonlin} ignoring uncertainty in the $f_i$ and $d_i$, as in \eqref{eq:detmap}.
\item[\texttt{linear}:] Same as \texttt{nonlin}, but forcing $\theta_{\sigma,1} = -\infty$ and hence linear $f_i$.
\item[\texttt{S-linear}:] Simplified version of \texttt{linear} ignoring uncertainty in the $f_i$ and $d_i$ as in \eqref{eq:detmap}, which results in a joint Gaussian posterior predictive distribution and is similar to the approach proposed and used in numerical comparisons in \citet{Kidd2020}.
\item[\texttt{DPM}:] The model with Dirichlet process mixture residuals described in Section \ref{sec:dpm}.
\item[\texttt{MatCov}:] Gaussian with zero mean and isotropic Mat\'ern covariance, whose three hyperparameters are inferred via maximum likelihood estimation.
\item[\texttt{tapSamp}:] Gaussian with a covariance matrix given by the sample covariance tapered (i.e., element-wise multiplied) by an exponential correlation matrix with range equal to the maximum pairwise distance among the locations.
\item[\texttt{autoFRK}:] resolution-adaptive automatic fixed rank kriging \citep{Tzeng2018,Tzeng2021} with approximately $\sqrt{N}$ basis functions.
\item[\texttt{local}:] a locally parametric method for climate data \citep{Wiens2021} that fits anisotropic Mat\'ern covariances in local windows and combines the local fits into a global model.
\end{description}
We also compared to a \texttt{VAE} \citep{Kingma2014} and a \texttt{GAN} designed for climate-model output \citep{Besombes2021}, but these deep-learning methods were not competitive in our simulation settings or for the climate data in Section \ref{sec:application} (see Appendix \ref{app:vae}).

We considered four simulation scenarios, for which samples are illustrated in the top row of Figure \ref{fig:maternsine}, consisting of a Gaussian distribution with an exponential covariance and three non-Gaussian extensions thereof. All scenarios can be characterized via transport maps as in Section \ref{sec:regression}, with $d_i$ as given by a Gaussian with exponential covariance in the form \eqref{eq:condmap}:
\begin{description}[itemsep=1pt,topsep=2pt,parsep=1pt]
    \item[LR900:] \textbf{L}inear map (i.e., a Gaussian distribution) with components $f_i^\text{L}(\by_{1:i-1}) = \sum_{k=1}^{i-1} b_{i,k} y_{c_i(k)}$, where the $b_{i,k}$ are based on an exponential covariance with unit variance and range parameter 0.3 on a \textbf{R}egular grid of size $N = 30 \times 30 = \mathbf{900}$ on the unit square.
    \item[NR900:] \textbf{N}onlinear extension of LR900 by a sine function of a weighted sum of the nearest two neighbors: $f_i^\text{NL}(\by_{1:i-1}) = f_i^\text{L}(\by_{1:i-1}) + 2 \sin(4(b_{i,1} y_{c_i(1)}+b_{i,2} y_{c_i(2)}))$ (see Figure \ref{fig:sinetrue})
    \item[NI3600:] Same as NR900, but at $N=\mathbf{3{,}600}$ \textbf{I}rregularly spaced locations sampled uniformly at random
    \item[NR900B:] Same as NR900, but with a \textbf{B}imodal distribution for the $\epsilon_i$ in \eqref{eq:regi}: $\epsilon_i | \mu_i, d_i \sim \normal(\mu_i,d_i^2)$ with $\mu_i$ sampled from $\{-3.5 d_i, 3.5 d_i\}$ with equal probability
\end{description}
For computational simplicity, each (true) $f_i$ was assumed to only depend on the nearest 30 previously ordered neighbors, but this gives a highly accurate approximation of a ``full'' exponential covariance in the LR900 case, as the true fields exhibit strong screening due to being based on the same maximin ordering as our methods. A further ordering-invariant simulation scenario is considered in Appendix \ref{app:simprod}.

\begin{figure}
\centering
	\begin{subfigure}{.244\textwidth}
	\centering
	\includegraphics[trim=-22mm 0mm 7mm 0mm, clip,width =.98\linewidth]{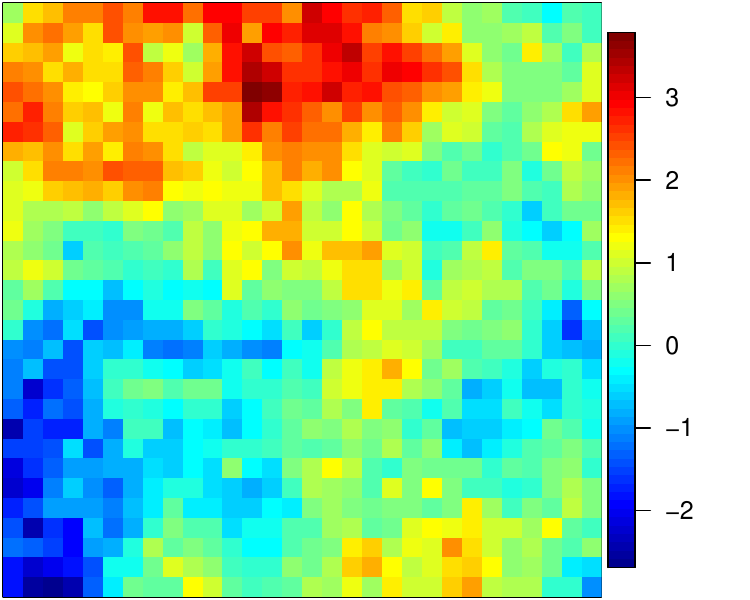} \\
	\includegraphics[width = 0.98\linewidth]{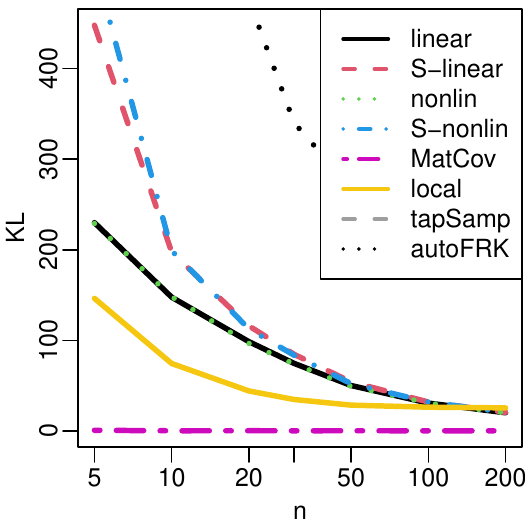}
	\caption{LR900}
	\label{fig:ms1}
	\end{subfigure}
	\hfill
	\begin{subfigure}{.244\textwidth}
	\centering
	\includegraphics[trim=-22mm 0mm 7mm 0mm, clip,width =.98\linewidth]{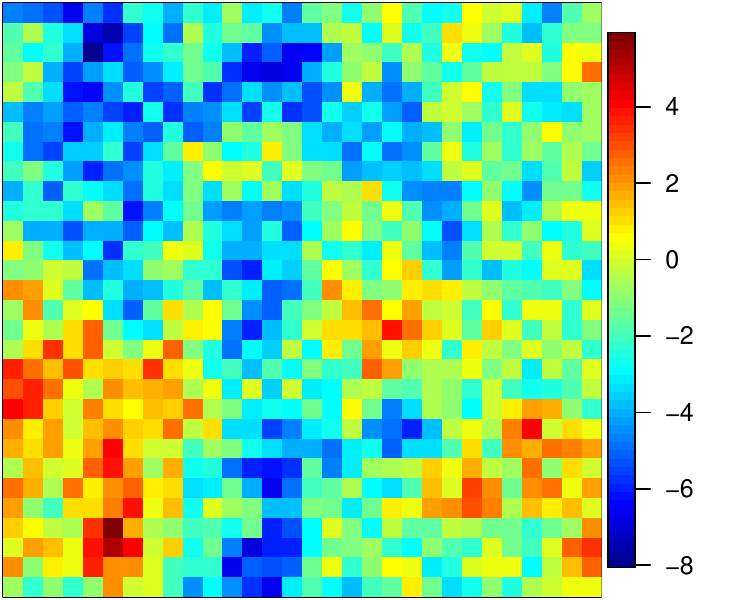} \\
 	\includegraphics[width = 0.98\linewidth]{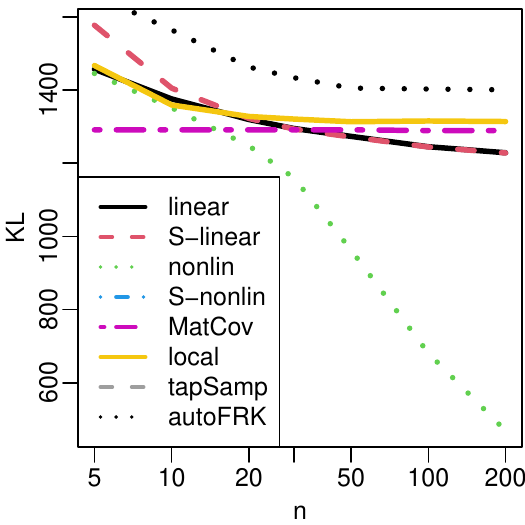}
 	\caption{NR900}
	\label{fig:ms2}
	\end{subfigure}
	\hfill
	\begin{subfigure}{.244\textwidth}
	\centering
	 \includegraphics[trim=-22mm 0mm 7mm 0mm, clip,width =.98\linewidth]{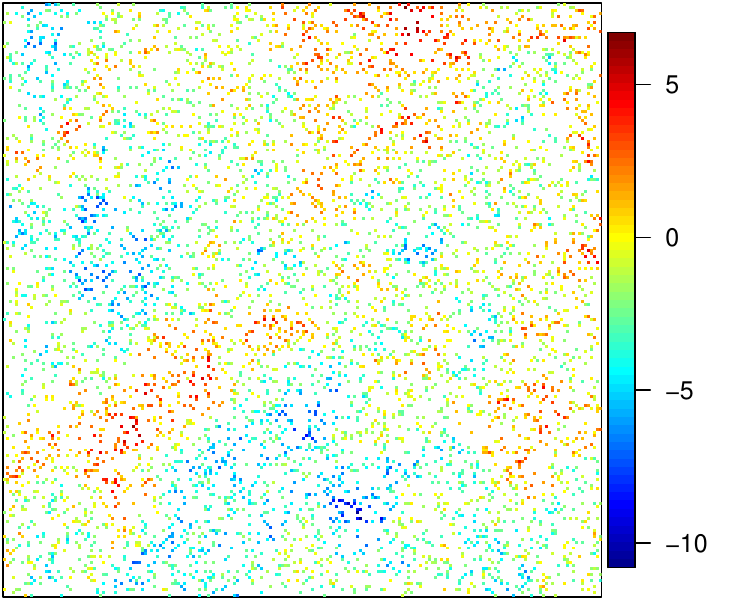} \\
 	\includegraphics[width = 0.98\linewidth]{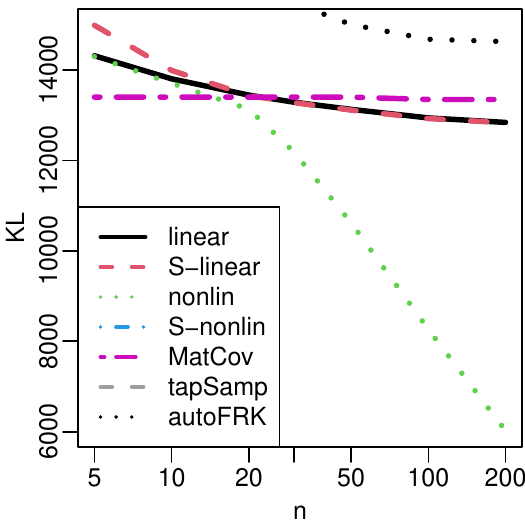}
 	\caption{NI3600}
	\label{fig:ms3}
	\end{subfigure}
	\hfill
	\begin{subfigure}{.244\textwidth}
	\centering
	\includegraphics[trim=-22mm 0mm 7mm 0mm, clip,width =.98\linewidth]{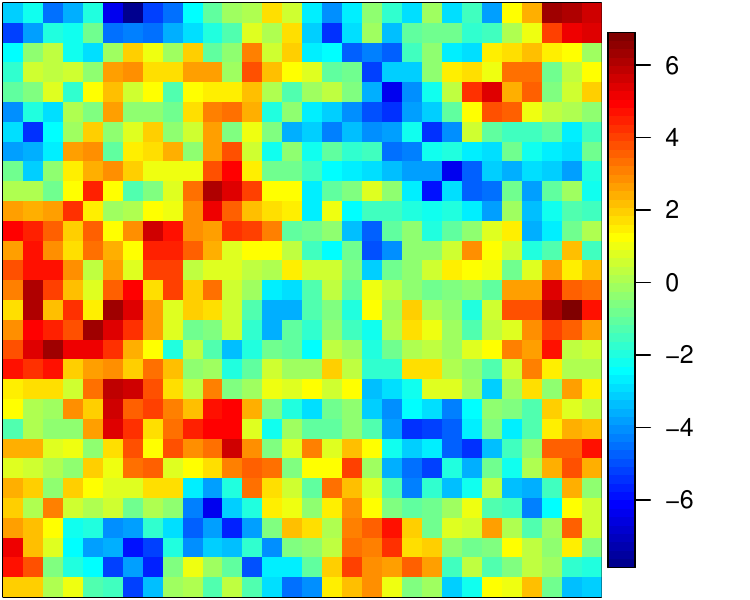} \\
 	\includegraphics[width = 0.98\linewidth]{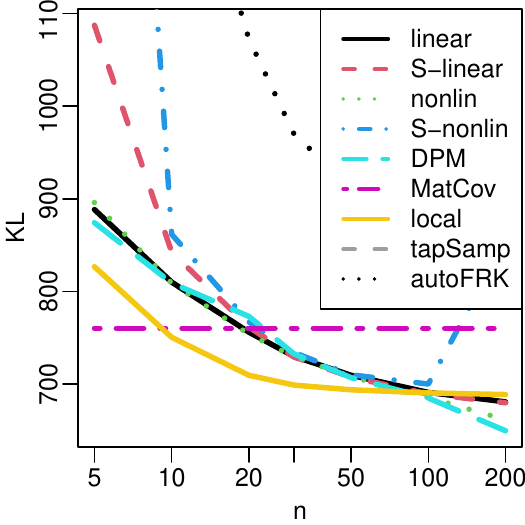}
 	\caption{NR900B}
	\label{fig:ms4}
	\end{subfigure}
  \caption{Top row: Simulated spatial fields for four simulation scenarios described in Section \ref{sec:simstudy}. Bottom row: Corresponding comparisons of KL divergence as a function of ensemble size $n$ (on a log scale) for different methods. The KL divergences for \texttt{tapSamp} in (a)--(d) and for \texttt{S-nonlin} in (b)--(c) were too high and are not visible. \texttt{DPM} is only included in (d), while \texttt{local} is omitted from (c) because it was created for regular grids.}
\label{fig:maternsine}
\end{figure}

We compared the accuracy of the methods via the Kullback-Leibler (KL) divergence,
\begin{equation}
\label{eq:kl}
E( \log p_0(\by) )  - E(\log p(\by|\bY) ),
\end{equation}
between the true distribution $p_0(\by)$ and the inferred distribution $p(\by|\bY)$ implied by the posterior map (see \eqref{eq:tstar}), where the expectations are taken with respect to the true distribution. We approximated the expectations by averaging over 50 simulated test fields $\by^\star$, and so the resulting KL divergence is the difference of the log-scores \citep[e.g.,][]{Gneiting2014} of the true and inferred distributions. 

The results are shown in Figure \ref{fig:maternsine}. Whenever nonlinear structure was not discernible from the data (because the true map was linear or because the ensemble size $n$ was too small), \texttt{nonlin} performed similarly to \texttt{linear} and hence did not suffer due to its over-flexibility. For larger ensemble size and nonlinear truths, \texttt{nonlin} at times far outperformed \texttt{linear}. \texttt{S-linear} and \texttt{S-nonlin} were generally less accurate than their counterparts with uncertainty quantification; in the linear LR900 setting, this was only an issue for small ensemble size, but \texttt{S-nonlin} performed extremely poorly when the nonlinear structure was clearly apparent in the data, likely due to overfitting without accounting for uncertainty. 
\texttt{tapSamp} and \texttt{autoFRK} performed uniformly worst.
As \texttt{MatCov} (with smoothness 0.5) is the true model for LR900, it was almost exact in that scenario. The other three scenarios are extensions of a Mat\'ern GP, and so \texttt{MatCov} also performed well for $n<20$ or so. The \texttt{local} Mat\'ern method was less accurate than \texttt{MatCov} for LR900 and NR900 but performed well for NR900B. For simulation scenarios that deviate more strongly from a Mat\'ern GP, \texttt{nonlin} was uniformly more accurate than \texttt{MatCov} and \texttt{local} (see Appendix \ref{app:simprod}).

Estimating $\bftheta$ via stochastic gradient ascent with 3 epochs and fitting the map based on $n=20$ samples took less than 7 seconds for the scenarios with $N=900$ and less than 44 seconds for the larger NI3600 scenario for \texttt{nonlin}, \texttt{linear}, \texttt{S-linear}, and \texttt{S-nonlin} on a single core on a laptop (2.5 GHz Intel Core i7 with 16GB RAM); \texttt{DPM} required a total of around 16 minutes for 500 MCMC iterations for NR900B.

\section{Climate-data application \label{sec:application}}

An important application of our methods is the analysis and emulation of output from climate models. Climate models are essentially large sets of computer code describing the behavior of the Earth system (e.g., the atmosphere) via systems of differential equations. Much time and resources have been spent on developing these models, and enormous computational power is required to produce ensembles (i.e., solve the differential equations for different starting conditions) on fine latitude-longitude grids for various scenarios of greenhouse-gas emissions. Of the large amount of data and output that have been generated, only a small fraction has been fully explored or analyzed \citep[e.g.][]{Benestad2017}.
Stochastic weather generators infer the distribution of one or more variables, so that relevant summaries or additional samples can be computed more cheaply than via more runs of the computer model.

\begin{figure}
\centering
\begin{minipage}{.94\textwidth}
\centering
	\begin{subfigure}{.495\linewidth}
	\centering
 	\includegraphics[width =.99\linewidth]{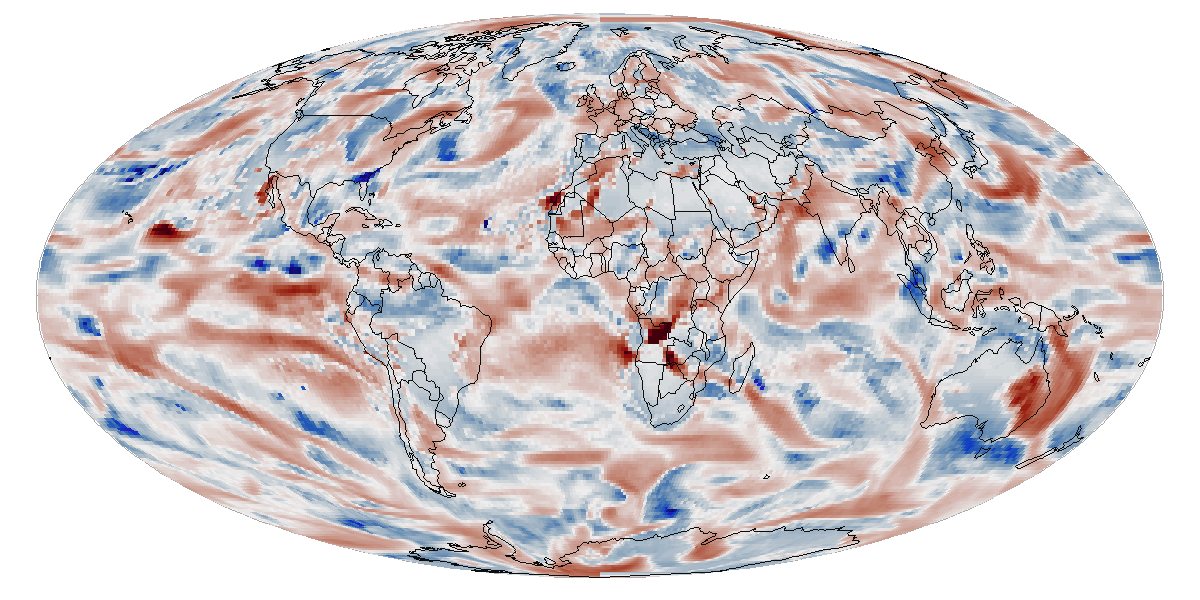}
	\end{subfigure}%
	\begin{subfigure}{.495\linewidth}
	\centering
 	\includegraphics[width =.99\linewidth]{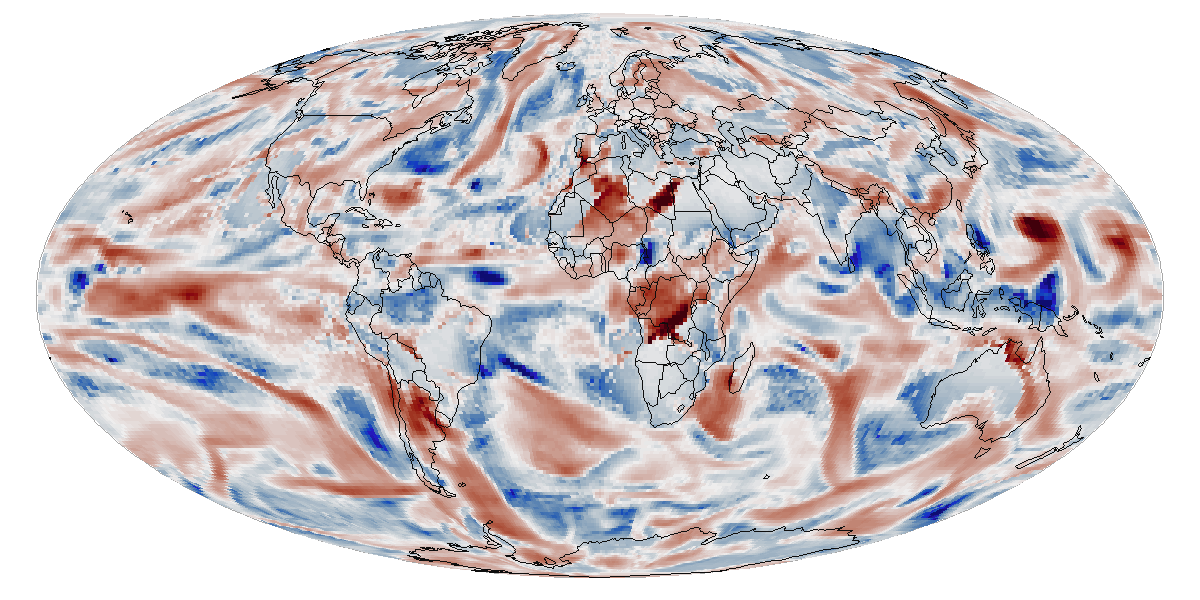}
	\end{subfigure}%
\end{minipage}
\hfill
\begin{minipage}{.04\textwidth}
%
\includegraphics[trim=400mm 0mm 0mm 0mm, clip,width =.55\linewidth]{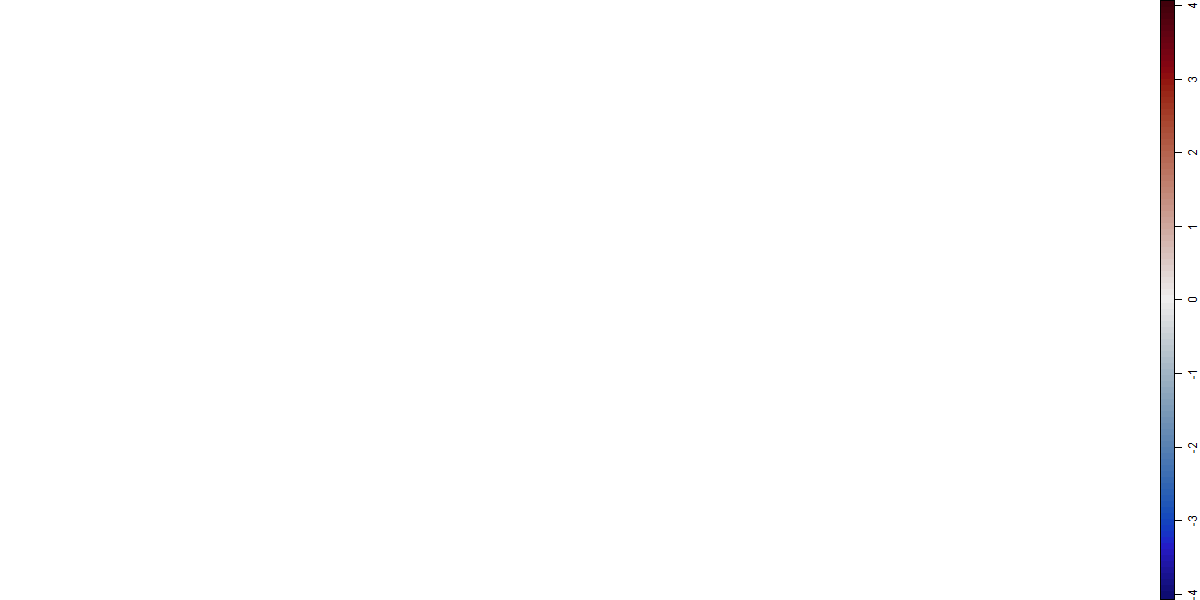} 
\end{minipage}	
  \caption{Two members of an ensemble of log-transformed precipitation anomalies produced by a climate model, on a global grid of size $N = 288 \times 192 = 55{,}296$. We want to infer the underlying $N$-dimensional distribution based on an ensemble of $n<100$ training samples.}
\label{fig:precdata}
\end{figure}

We considered log-transformed total precipitation rate (in m/s) on a roughly $1^\circ$ longitude-latitude global grid of size $N = 288 \times 192 = 55{,}296$ in the middle of the Northern summer (July 1) in $98$ consecutive years (the number of years contained in one NetCDF data file), starting in the year 402, from the Community Earth System Model (CESM) Large Ensemble Project \citep{Kay2015}.
We obtained precipitation anomalies by standardizing the data at each grid location to mean zero and variance one, shown in Figure \ref{fig:precdata}.
For our methods, we used chordal distance to compute the maximin ordering and nearest neighbors. 

For ease of comparison and illustration, we first considered a smaller grid of size $N = 37 \times 74 = 2{,}738$ in a subregion containing large parts of the Americas ($45^\circ$S to $45^\circ$N and $130^\circ$W to $30^\circ$W) containing ocean, land, and mountains.
As shown in Figure \ref{fig:precipprop}, the precipitation anomalies exhibited similar features as our simulated data in Figure \ref{fig:matern}, with regression data concentrating on lower-dimensional manifolds and weights decaying rapidly as a function of neighbor number.

\begin{figure}
\centering
~
\hfill	
	\begin{subfigure}{.36\textwidth}
	\centering
  	\includegraphics[trim=4mm 13mm 9mm 22mm, clip,width =.99\linewidth]{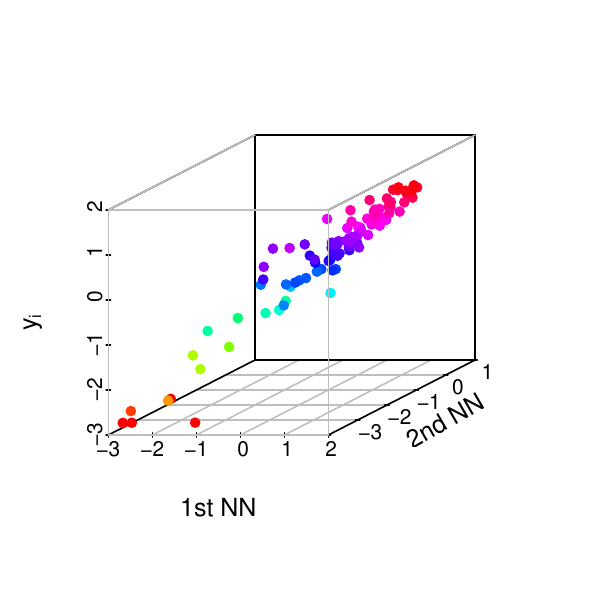}
	\caption{$\by_{2000}$ vs 1st and 2nd NN}
	\label{fig:pre3}
	\end{subfigure}%
\hfill	
	\begin{subfigure}{.264\textwidth}
	\centering
 	\includegraphics[width =.99\linewidth]{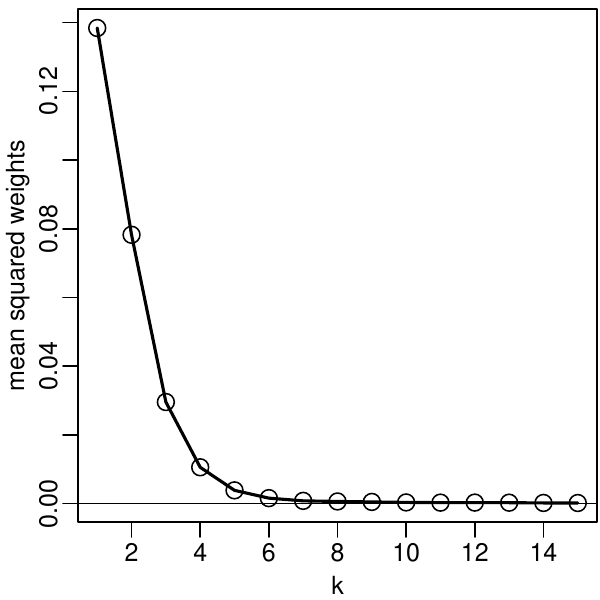}
	\caption{Avg.\ $\{\hat{b}_{i,k}^2: i\!=\!1,\ldots,N\}$}
	\label{fig:precweights}
	\end{subfigure}%
\hfill 
~
	\caption{
	The precipitation anomalies (in the Americas subregion) have similar properties as the Gaussian distribution with exponential covariance in Figure \ref{fig:matern}:
	(a) Our approach can be viewed as $N$ regressions as in \eqref{eq:regi} of each $y_i$ on ordered nearest neighbors (NNs), with the regression data on low-dimensional manifolds. 
    (b) For linear regressions with $f_i(\by_{1:i-1}) = \sum_{k=1}^{i-1} y_{c_i(k)}b_{i,k}$ fitted via Lasso, the squared (estimated) regression coefficients decay rapidly as a function of neighbor number $k$.}
\label{fig:precipprop}
\end{figure}

\begin{figure}
\centering
~
\hfill	
	\begin{subfigure}{.264\textwidth}
	\centering
 	\includegraphics[width =.99\linewidth]{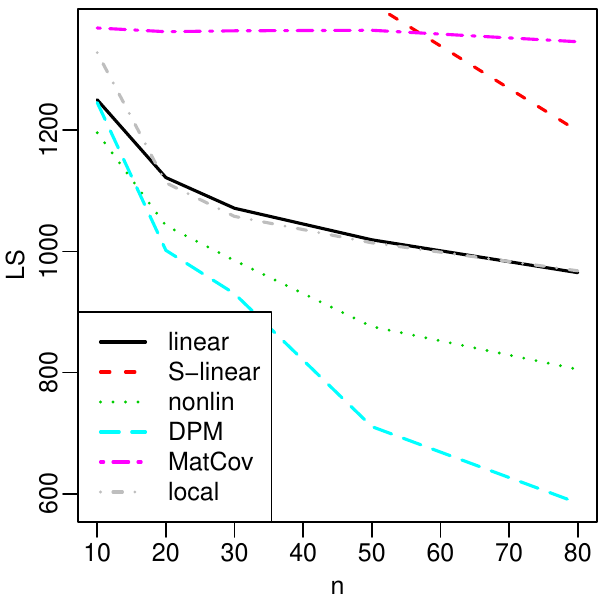}
	\caption{LS (Americas)}
	\label{fig:pre4}
	\end{subfigure}%
\hfill		
	\begin{subfigure}{.264\textwidth}
	\centering
 	\includegraphics[width =.99\linewidth]{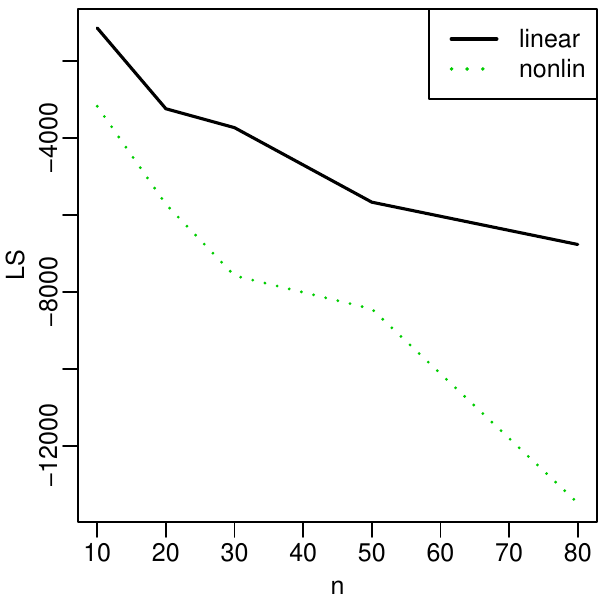}
	\caption{LS (global)}
	\label{fig:global2}
	\end{subfigure}%
\hfill 
~
\caption{
For precipitation anomalies, comparison of log-score (LS; equal to KL divergence up to an additive constant) for estimated joint distribution as a function of ensemble size $n$:
(a) Americas subregion; \texttt{S-nonlin}, \texttt{tapSamp}, and \texttt{autoFRK} are not shown because their LS were too high.
(b) LS for \texttt{linear} and \texttt{nonlin} for precipitation anomalies on the global grid.}
\label{fig:precipcomp}
\end{figure}

For comparing the methods from Section \ref{sec:simstudy} on the precipitation anomalies, computing the KL divergence as in \eqref{eq:kl} was not possible, as the true distribution $p_0(\by)$ was unknown. Hence, we compared the methods using various training data sizes $n$ in terms of log-scores, which approximate the KL divergence up to an additive constant; specifically, these log-scores consist of the second part of \eqref{eq:kl}, $- E(\log p(\by|\bY) )$, with the expectation approximated by averaging over 18 test replicates and over five random training/test splits. 

The comparison for the Americas subregion is shown in Figure \ref{fig:pre4}. (A prediction comparison for partially observed test data provided in Appendix \ref{app:pred} produced similar results.) \texttt{nonlin} outperformed \texttt{linear}, and \texttt{DPM} was even more accurate than \texttt{nonlin} for large $n$, indicating that the precipitation anomalies exhibit joint and marginal non-Gaussian features. 
As in Section \ref{sec:simstudy}, \texttt{S-linear} and \texttt{S-nonlin} performed poorly due to ignoring uncertainty in the estimated map. \texttt{local} performed similarly to \texttt{linear} but was less accurate than \texttt{nonlin} and \texttt{DPM} for all $n$.
\texttt{VAE}, \texttt{MatCov}, \texttt{tapSamp}, and \texttt{autoFRK} were not competitive for this dataset.

We also considered the map coefficients $\bz^\star = \pmap(\by^\star)$ discussed in Sections \ref{sec:invertiblemap} and \ref{sec:spatialinference}, using the map obtained by fitting \texttt{nonlinear} to the first $n=97$ replicates as training data.
In Figure \ref{fig:precip_coef}, the map coefficients for a held-out test field appeared roughly i.i.d.\ standard Gaussian, with sample autocorrelations near zero (not shown).
Figure \ref{fig:precipcoef} illustrates that the map coefficients offer similar properties for non-Gaussian fields as principal-component scores do for Gaussian settings. For example, the medians of the posterior distributions of the $d_i$ (see \eqref{eq:nigpost}) decreased rapidly as a function of $i$, which means that the map coefficients early in the maximin ordering captured much more (nonlinear) variation than later-ordered coefficients (see, e.g., \eqref{eq:invmap} and \eqref{eq:detmap}). Further, we computed the map coefficients for all 98 replicates for July 2--30 (still based on the posterior map trained on July 1 data), and the lag-1 autocorrelation over time between map coefficients also decreased with $i$. Specifically, while most of the first 100 were greater than 0.2, many later autocorrelations were negligible; this indicates that a spatio-temporal analysis could proceed by fitting a simple (linear) autoregressive model over time to only the first $k$, say, map coefficients, while treating the remaining $N-k$ coefficients as independent over time.
As shown in Appendix \ref{app:recon}, the \texttt{nonlinear} map coefficients strongly outperformed standard linear principal components in terms of dimension reduction and reconstruction of the Americas climate fields.

\begin{figure}
\centering
~
\hfill	
	\begin{subfigure}{.26\textwidth}
	\centering
  	\includegraphics[width =.95\linewidth]{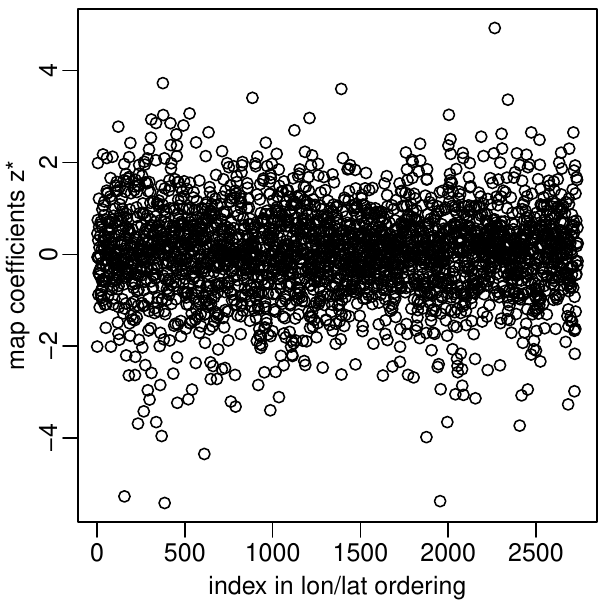}
	\caption{Test map coefficients $\bz^\star$}
	\label{fig:precip_coef}
	\end{subfigure}%
\hfill	
	\begin{subfigure}{.26\textwidth}
	\centering
  	\includegraphics[width =.95\linewidth]{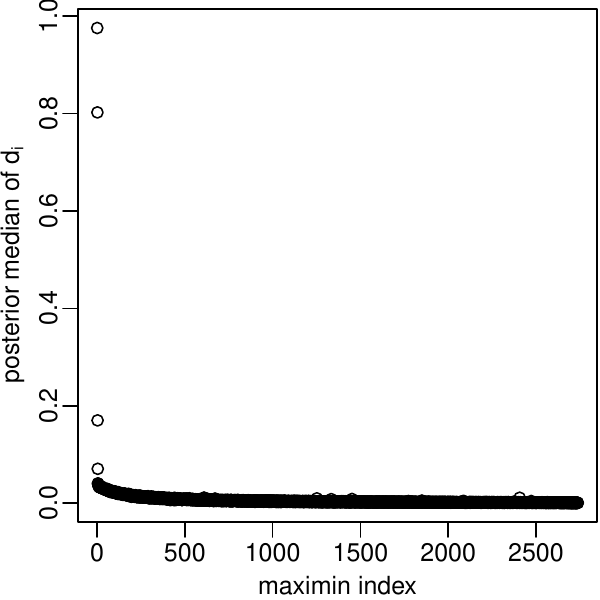}
	\caption{Posterior median of $d_i$}
	\label{fig:precip_di}
	\end{subfigure}%
\hfill	
	\begin{subfigure}{.26\textwidth}
	\centering
 	\includegraphics[width =.95\linewidth]{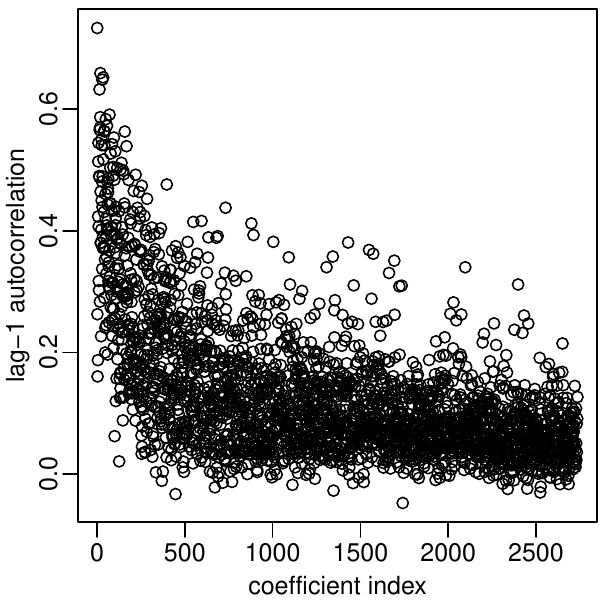}
	\caption{Lag-1 autocorrelation}
	\label{fig:precip_acf}
	\end{subfigure}%
\hfill 
~
	\caption{Properties of the map coefficients $\bz^\star = \pmap(\by^\star)$ for the precipitation anomalies on the grid of size $N=2{,}738$ in the Americas subregion.
	(a): The map coefficients corresponding to the test field in Figure \ref{fig:prec1} in the original data ordering (first by longitude, then latitude) appeared roughly i.i.d.\ standard Gaussian, aside from slightly heavy tails.
	(b): The posterior medians of the $d_i$ decreased rapidly as a function of $i$, meaning that the first few map coefficients captured much more variation than later-ordered coefficients. 
	(c): The autocorrelation between consecutive days also decreased with $i$; while most were greater than 0.2 for $i<100$, many autocorrelations for $i>100$ were negligible.}
\label{fig:precipcoef}
\end{figure}

\begin{figure}
\centering
\begin{minipage}{.94\textwidth}
\centering
	\begin{subfigure}{.495\linewidth}
	\centering
 	\includegraphics[width =.99\linewidth]{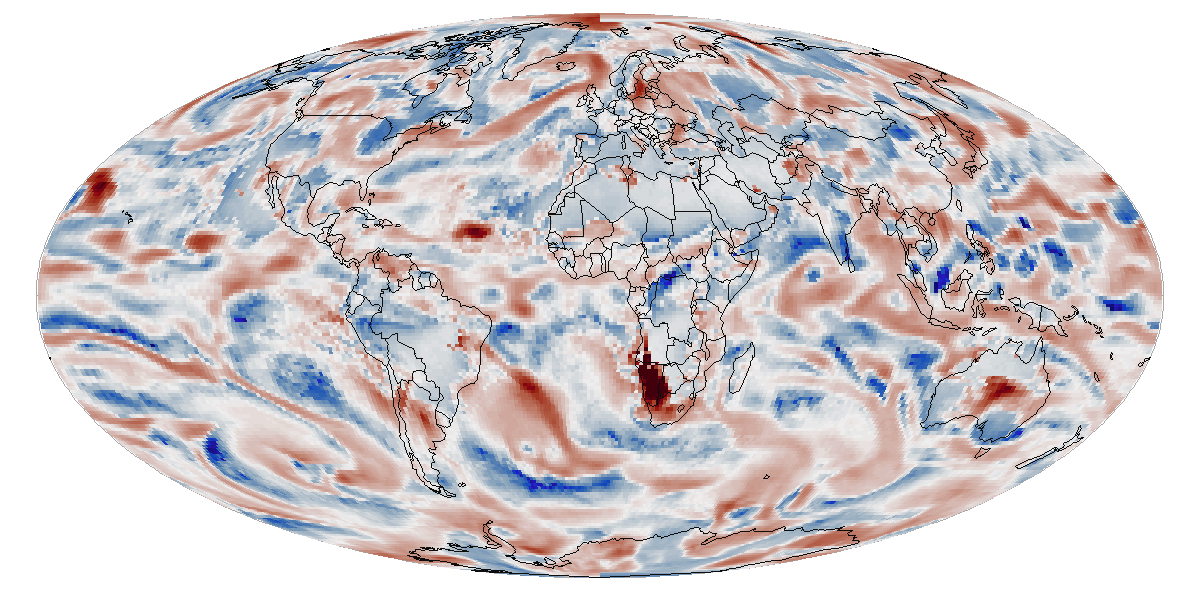}
 	\caption{Test field ($i=N$)}
 	\label{fig:prec1}
	\end{subfigure}%
	\begin{subfigure}{.495\linewidth}
	\centering
 	\includegraphics[width =.99\linewidth]{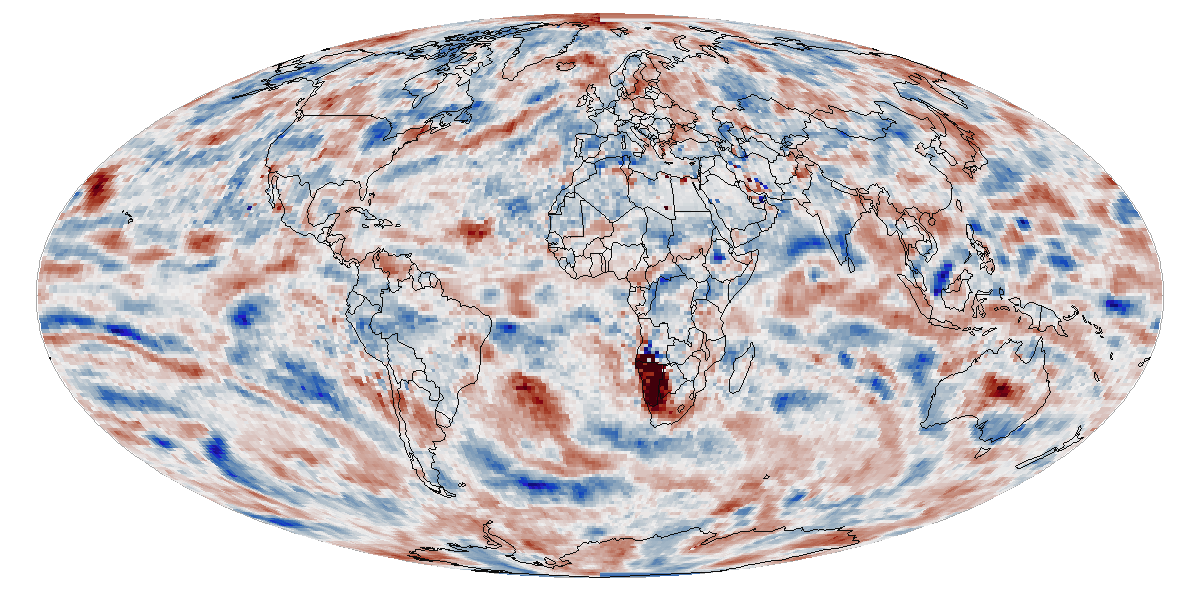}
 	\caption{$i=5{,}000$}
	\end{subfigure}%
	\\
	\smallskip
	\begin{subfigure}{.495\linewidth}
	\centering
 	\includegraphics[width =.99\linewidth]{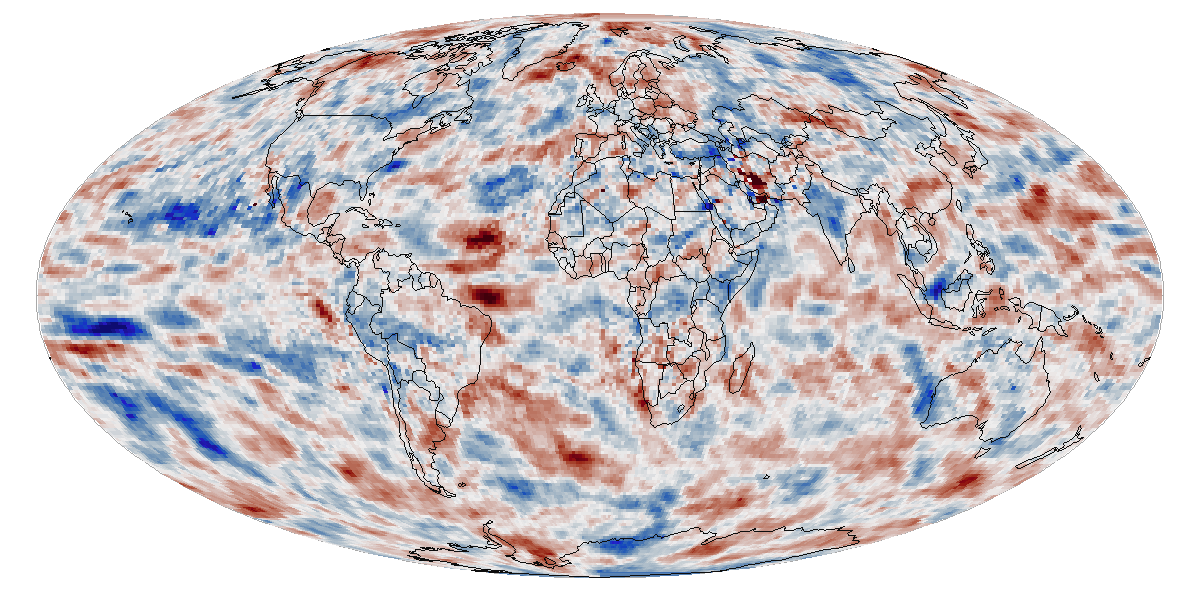}
 	\caption{$i=500$} 	
	\end{subfigure}%
	\begin{subfigure}{.495\linewidth}
	\centering
 	\includegraphics[width =.99\linewidth]{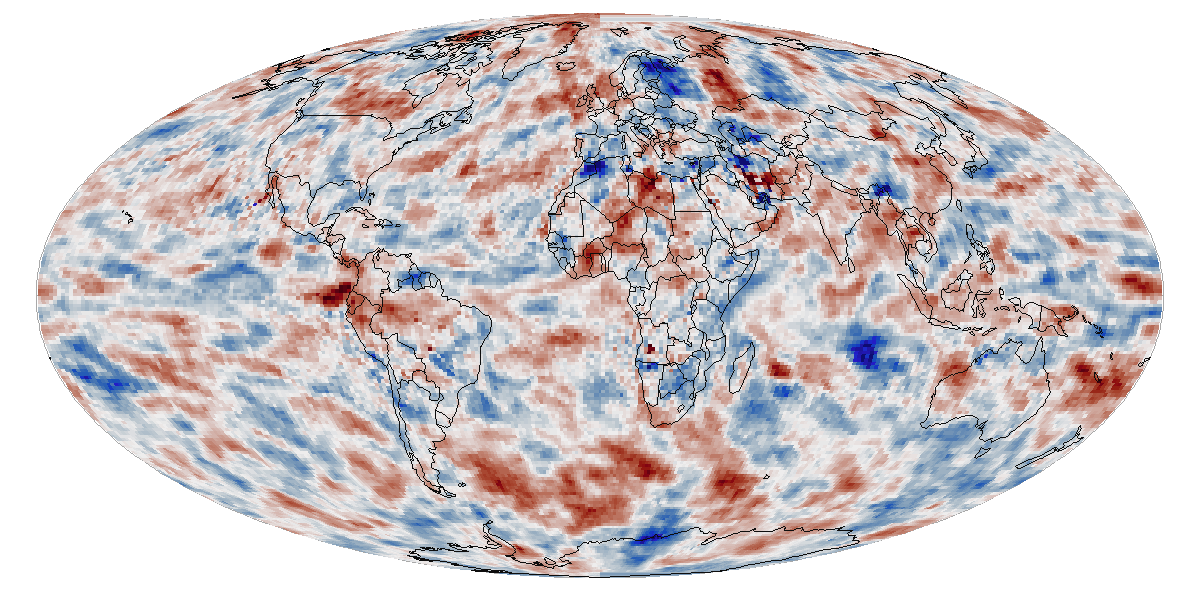}
 	\caption{Unconditional ($i=0$)}
 	\label{fig:prec4}
	\end{subfigure}%
\end{minipage}
\hfill
\begin{minipage}{.05\textwidth}
%
\includegraphics[trim=400mm 0mm 0mm 0mm, clip,width =.98\linewidth]{precip_global_legend.png} 
\end{minipage}	
  \caption{For the global climate data ($N=55{,}296$), we fitted a stochastic emulator using \texttt{nonlin} based on $n=97$ training replicates. Given a held-out test field $\by^*$ in (a), we show conditional simulations based on fixing the first $i$ map coefficients in $\bz^* = \pmap(\by^\star)$. (b): Only differs in some fine-scale features from (a). (c): Some large-scale features from (a) are preserved. (d): Unconditional simulation (i.e., independent from (a)).}
\label{fig:precipcond}
\end{figure}

To demonstrate scalability to large datasets, we compared \texttt{linear} and \texttt{nonlinear} on the entire global precipitation anomaly fields of size $N = 288 \times 192 = 55{,}296$. As shown in Figure \ref{fig:global2}, \texttt{nonlin} outperformed \texttt{linear} even more decisively than for the Americas subregion. Even in the largest and most accurate setting ($n=80$), the estimated $\bftheta$ for \texttt{nonlin} implied $m=9$, meaning that the corresponding transport maps were extremely sparse and hence computationally efficient; estimating $\bftheta$ (4 epochs) and fitting the map for \texttt{nonlin} took only around 6 minutes on a single core on a laptop (2.5 GHz Intel Core i7 with 16GB RAM) for $n=10$. 
In contrast, \texttt{MatCov} and \texttt{local} (which already took about two hours for the much smaller Americas region) were too computationally demanding for the global data. A Vecchia approximation of \texttt{MatCov} resulted in a log-score above +77,000 and was thus not competitive.
Also, for \texttt{nonlin} all but 113 of the $N=55{,}296$ posterior medians of the $d_i$ were more than 20 times smaller than the largest posterior median (i.e., that of $d_1$), indicating that our approach could be used for massive dimension reduction without losing too much information.

Finally, the fitted map (or rather, its inverse $\pmap^{-1}$) can also be viewed as a stochastic emulator of the climate model. Specifically, we can produce a new precipitation-anomaly sample by drawing $\bz^* \sim \normal_N(\bfzero,\bI_N)$ and then computing $\by^* = \pmap^{-1}(\bz^*)$. One such sample (for the full global grid) is shown in Figure \ref{fig:prec4} and appears qualitatively similar to the model output in Figure \ref{fig:precdata}; while producing the latter requires a supercomputer, the former can be generated in a few seconds on a laptop. 
Further, our approach can also be used to draw conditional samples, in which we fix the first $i$, say, map coefficients, for example at the values corresponding to a given spatial field. Such draws, which maintain the large-scale features in the held-out (98th) test field but allow for newly sampled fine-scale features, are shown in Figure \ref{fig:precipcond}. 
This is related to the supervised conditional sampling ideas in \citet{Kovachki2020}, with their inputs given by our first $i$ ordered test observations.

\section{Conclusions\label{sec:conclusions}}

We have developed a Bayesian approach to inferring a non-Gaussian target distribution via a transport map from the target to a standard normal distribution. The components of the map are modeled using Gaussian processes. For the distribution of spatial fields, we have developed specific prior assumptions that result in sparse maps and thus scalability to high dimensions. 
Instead of manually or iteratively expanding a finite-dimensional parameterization of the transport map, our Bayesian approach probabilistically regularizes the map; the resulting approach is flexible and nonparametric, but guards against overfitting and quantifies uncertainty in the estimation of the map. 
Because our method can be fitted rapidly, is fully automated, and was highly accurate in our numerical comparisons, we recommend it for most spatial emulation tasks, except for applications in which very few replicates are available or for which exploratory analyses have shown that a (Gaussian) parametric approach can provide a good fit.
In addition, due to conjugate priors and the resulting closed-form expressions for the posterior map and its inverse, our approach also allows us to convert non-Gaussian data into i.i.d.\ Gaussian map coefficients, which can be thought of as a nonlinear extension of principal components.

As our approach essentially turns estimation of a high-dimensional joint distribution into a series of GP regressions, it is straightforward to include additional covariates and examine their nonlinear, non-Gaussian effect on the distribution.
Shrinkage toward a joint Gaussian distribution with a parametric covariance function could be achieved by assuming the mean for the GP regressions to be the one implied by a Vecchia approximation of that covariance function \citep{Kidd2020}; this could enable meaningful predictions at unobserved spatial locations (cf.\ Appendix \ref{app:pred}).
Extensions to more complicated input domains (e.g., space-time) could be obtained using correlation-based ordering (Section \ref{sec:maximin}). 
Another major avenue of future work would be to use the inferred distribution as the prior of a latent field, which we then update to obtain a posterior given noisy observations; among numerous other applications, this would enable the use of our technique to infer the forecast distribution and account for uncertainty in ensemble-based data assimilation \citep{Boyles2020}, leading to nonlinear updates for non-Gaussian applications. 
We are currently pursuing multiple extensions and applications of our methods to climate science, including climate-change detection and attribution, climate-model calibration, and climate-model emulation and interpolation in covariate space (e.g., as a function of CO$_2$ emissions).

\footnotesize
\appendix
\section*{Acknowledgments}

Katzfuss was partially supported by National Science Foundation (NSF) Grants DMS--1654083, DMS--1953005, and CCF--1934904, and by NASA's Advanced Information Systems Technology Program (AIST-21). Sch{\"a}fer gratefully acknowledges support by the Air Force Office of Scientific Research under award number FA9550-18-1-0271, and the Office of Naval Research under award N00014-18-1-2363. We would like to thank Joe Guinness and several reviewers for helpful comments. We are especially grateful to Jian Cao, who wrote a Python implementation, produced timing results, and obtained GAN and VAE results, and to Trevor Harris, who created the VAE implementation for our numerical comparisons.





\section{Proofs\label{app:proofs}}

\begin{proof}[Proof of Proposition \ref{prop:maps}]
Combining \eqref{eq:map2norm} with the conditional independence of $\by^{(1)},\ldots,\by^{(n)}$, we have
\begin{equation}
    \label{eq:reg}
\textstyle p(\bY|\bf,\bd) = \prod_{i=1}^N \prod_{j=1}^n \normal(y_i^{(j)}|f_i(\by_{1:i-1}^{(j)}),d_i^2) = \prod_{i=1}^N \normal_n(\by_i|\bf_i,d_i^2\bI_n),
\end{equation}
where $\bf_i = f_i(\bY_{1:i-1}) = \big(f_i(\by_{1:i-1}^{(1)}),\ldots,f_i(\by_{1:i-1}^{(n)}) \big)^\top$ is distributed as $\bf_i|d_i,\bY_{1:i-1} \sim \normal(\bfzero,d_i^2\bK_i)$. 
Combined with \eqref{eq:dprior}, we see that $\bf_i, d_i$ (conditional on $\bY_{1:i-1}$) jointly follow a (multivariate) normal-inverse-gamma (NIG) distribution, independently for each $\bf_i,d_i$. 
Given the data $\bY$ as in \eqref{eq:reg}, well-known conjugacy results imply that the posterior of $\bF = (\bf_1,\ldots,\bf_N)$ and $\bd=(d_1,\ldots,d_N)$ also consists of independent NIG distributions:
\begin{equation}
\textstyle p(\bF,\bd|\bY)  
\propto \prod_{i=1}^N p(\by_i|\bf_i,d_i) p(\bf_i|d_i,\bY_{1:i-1}) \, p(d_i) \propto \prod_{i=1}^N \normal(\bf_i|\hat\bf_i,d_i^2 \tilde\bK_i) \, \mathcal{IG}(d_i^2|\tilde\alpha_i,\tilde\beta_i), \label{eq:nigpost}
\end{equation}
where $\tilde\bK_i = \bK_i - \bK_i \bG_i^{-1}\bK_i$ and $\hat\bf_i = \bK_i \bG_i^{-1} \by_i$.

We have
$p(\by^\star|\bf, \bd) = \prod_{i=1}^N \normal(y_i^\star|f_i(\by_{1:i-1}^*),d_i^2)$ using \eqref{eq:map2norm}. Combining this with \eqref{eq:nigpost} and the conditional-independence assumption in \eqref{eq:gp}, the posterior predictive distribution can be shown to be
\[
p(\by^\star|\bY) = \prod_{i=1}^N \int p(y_i^\star|\by_{1:i-1}^*,\bY,d_i) p(d_i|\bY) d d_i,
\]
where basic GP regression implies
\begin{equation}
\label{eq:condpred}
y_i^\star|\by^\star_{1:i-1},\bY,d_i \sim \normal\big(\hat f_i(\by^\star_{1:i-1}), d_i^2( v_i(\by^\star_{1:i-1})+1 )\big), \qquad i=1,\ldots,N.
\end{equation}
Combining this with $d_i^2 | \bY \sim \mathcal{IG}(\tilde\alpha_i,\tilde\beta_i)$ from \eqref{eq:nigpost},
we obtain the posterior predictive distribution as a product of $t$ densities,
\begin{equation}
    \label{eq:tstar}
    \textstyle p(\by^\star|\bY) = \prod_{i=1}^N t_{2\tilde\alpha_i}\big(y_i^\star\big|\hat f_i(\by^\star_{1:i-1}), \hat d_i^2( v_i(\by^\star_{1:i-1})+1 )\big),
\end{equation}
where our notation is such that $w \sim t_{\kappa}(\mu,\sigma^2)$ implies that $(w - \mu)/\sigma$ follows a ``standard'' $t$ with $\kappa$ degrees of freedom.
Hence, $\hat d_i^{-1} (v_i(\by^\star_{1:i-1})+1)^{-1/2}(y_i - \hat f_i(\by^\star_{1:i-1}))$ follows a $t_{2\tilde\alpha_i}$ distribution.
Using the fact that we can map from a distribution to the standard uniform using its cumulative distribution, the transformation
$\bz^\star = \pmap(\by^\star) \sim \normal_N(\bfzero,\bI_N)$ to a standard normal can be described using a triangular map with components
\begin{equation}
\label{eq:singlemap2}
z_i^\star = \pmap_i(y_1^\star,\ldots,y_i^\star) = \Phi^{-1}\big( F_{2\tilde\alpha_i}\big( \hat d_i^{-1} (v_i(\by^\star_{1:i-1})+1)^{-1/2}(y_i^\star - \hat f_i(\by^\star_{1:i-1})) \big)\big).
\end{equation}
The solution $\by^\star$ to the nonlinear triangular system $\pmap(\by^\star) =\bz^\star$ is found recursively by solving \eqref{eq:singlemap2} for $y_i^\star$:
\begin{equation}
y_i^\star = F_{2\tilde\alpha_i}^{-1}(\Phi(z_i^\star))\, \hat d_i (v_i(\by_{1:i-1}^\star)+1)^{1/2} + \hat f_i(\by_{1:i-1}^\star).
\end{equation}
\end{proof}

\begin{proof}[Proof of Proposition \ref{prop:lik}]
From \eqref{eq:gp}, we have that $\bf_i | d_i \stackrel{ind.}{\sim} \normal_n(\bfzero,d_i^2\bK_i)$; together with \eqref{eq:reg}, this implies that $\by_i|d_i,\bY_{1:i-1} \stackrel{ind.}{\sim}  \normal_n(\bfzero,d_i^2\bG_i)$. Combining this with \eqref{eq:dprior}, it is well known that $\by_i | \bY_{1:i-1} \stackrel{ind.}{\sim} t_{2\alpha_i}(\bfzero,\frac{\beta_i}{\alpha_i}\bG_i)$, where we define a multivariate $t$ distribution such that $\bw \sim t_{\kappa}(\bfmu,\bfSigma)$ implies that the entries of $\bfSigma^{-1/2}(\bw - \bfmu)$ are i.i.d.\ standard $t$ with $\kappa$ degrees of freedom. 
Plugging in the $t$ densities and simplifying using $\tilde\alpha_i = \alpha_i + n/2$, 
$\tilde\beta_i = \beta_i + \by_i{}^\top \bG_i^{-1} \by_i/2$, we can obtain
\begin{align}
  p(\bY) 
  & = \textstyle \prod_{i=1}^N t_{2\alpha_i}(\by_i |\bfzero,\frac{\beta_i}{\alpha_i}\bG_i)\\
  & \propto \textstyle \prod_{i=1}^N \Gamma(\tilde\alpha_i) \big( \Gamma(\alpha_i) (\alpha_i \beta_i/\alpha_i)^{n/2} |\bG_i|^{1/2}\big)^{-1} \big( 1+ \alpha_i/(\beta_i 2 \alpha_i) \by_i^\top \bG_i^{-1}\by_i \big)^{-\tilde\alpha_i}\\
  & \propto \textstyle\prod_{i=1}^N \big( \, |\bG_i|^{-1/2} \times ({\beta_i^{\alpha_i}}/{\tilde\beta_i^{\tilde\alpha_i}}) \times {\Gamma(\tilde\alpha_i)}/{\Gamma(\alpha_i)} \, \big),
\end{align}
where $\Gamma(\cdot)$ denotes the gamma function.
\end{proof}

\section{Conditional near-Gaussianity for quasiquadratic loglikelihoods\label{app:quasilinear}}

A Gaussian process with precision operator $A$ and mean $\mu$ has a quadratic negative loglikelihood given by $u \mapsto \frac{1}{2} \langle u, A u\rangle - \langle b, u\rangle$, with $b = A \mu$.
It is therefore closely related to the solution of systems of equation in $A$. 
For many popular smooth function priors, such as the Mat{\'e}rn process, the precision operator is a linear elliptic partial differential operator.
Just like minimizers of the quadratic energy $u \mapsto \frac{1}{2} \langle u, A u\rangle - \langle b, u\rangle$ satisfy the linear equation $Au = b$, the minimizers of quasiquadratic energies 
\begin{equation}
       \mathcal{E}(u) = \frac{1}{2}\langle u, L\left(D^ru\right)\rangle + V\left(D^{r-1}u, \ldots, u\right).
\end{equation}
with $L$ linear and $V$ possibly nonlinear, satisfy the quasilinear PDE 
\begin{equation}
    \label{eqn:quasilinear}
    L\left(D^ru\right) = - \frac{d}{du} V\left(D^{r-1}u, \ldots, u\right).  
\end{equation}

A natural stochastic model for phenomena governed by \eqref{eqn:quasilinear} is then the Boltzmann distribution at a finite temperature $T$; not concerning ourselves with the technical difficulties of a rigorous definition in the continuous case, the distribution is characterized by the likelihood
\begin{equation}
    \label{eqn:CHdensity}
    p(u) \propto \exp(- \mathcal{E}(u) / T).
\end{equation}

\Citet{cahn1958free} derived expressions of the form 
\begin{equation}
    \label{eqn:CHenergy}
    \mathcal{E}(u) = \int \nu |\nabla u(x)|^2 + (1 - u(x)^2)^2 dx
\end{equation}
for the free energy of a binary alloy, which has subsequently been applied to numerous other problems including multi-phase flows \citep{badalassi2003computation} and polymers \citep{choksi2009phase}.
We will use this energy and the associated Boltzmann distribution as an example to illustrate the conditional Gaussianity stochastic processes with quasiquadratic loglikelihoods.

In \eqref{eqn:CHenergy}, the field $u$ represents a mixture of two species.
The first, leading term of $\mathcal{E}(u)$ favors smooth functions by penalizing drastic jumps in concentration among nearby points.
The second term of $\mathcal{E}(u)$ describes a tendency of the two species to avoid mixing, favouring $u(x) = -1$ (mostly the first species) or $u(x) = 1$ (mostly the second species). 
When imposing the constraint that the overall abundance of the two species be equal ($\int u dx = 0$), minimizers of $\mathcal{E}$ need to carefully balance having values close to $\{-1, 1\}$ while also varying slowly in space, resulting in the formation of distinct positive or negative regions, the size of which is determined by the choice of $\nu$.
For each $x$, the distribution of $u(x)$ is then a mixture (in the probability-theoretic sense of the word) of the behavior of a positive region and that of a negative region, leading to a non-Gaussian marginal distribution.
If, however, we condition the process on averages over subdomains of decreasing diameter $\ell < 1$, we observe that the conditional distributions quickly become Gaussian.
The intuitive explanation for this phenomenon is that these averages contain enough information to determine, with high probability, whether a given point is part of a positive or negative region. 
Thus, its conditional distribution is dominated by the approximately Gaussian fluctuation around either a positive or negative value, as opposed to a mixture of these two distributions.
In Figure~\ref{fig:CHconditioning}, we used the preconditioned Crank-Nicholson proposal \citep{cotter2013mcmc} to simulate draws from such a Cahn-Hilliard process on a grid of $64 \times 64$ locations and plot the standardized histogram for a single location.
Our experimental results confirm our intuition that conditioning on averages over fine scales leads to increasingly Gaussian conditional distributions. 

This phenomenon can be understood in terms of the well-known Poincar{\'e} inequality \citep[Theorem 4.12]{adams2003sobolev}:
\begin{lemma}[Poincar{\'e} inequality]
    Let $\Omega \subset \mathbb{R}^d$ be a Lipschitz-bounded domain with diameter $\ell$, let $u$ and its first derivative be square-integrable, and let $u_{\Omega} = \frac{1}{|\Omega|} \int_{\Omega} u(x) dx $ be the mean of $u$ over $\Omega$.
    Then, we have 
    \begin{equation}
        \|u - u_{\Omega}\|_{L^2\left(\Omega\right)} \leq \ell \|\nabla u\|_{L^2\left(\Omega\right)}.
    \end{equation}
\end{lemma} 

The Poincar{\'e} inequality directly implies the following corollary.
\begin{corollary}
Let $\Omega$ be a Lipschitz-bounded domain. Let $\tau$ be a partition of $\Omega$ into Lipschitz-bounded subdomains with diameter upper bounded by $\ell$ and assume that $u,v$ and their first derivatives are square-integrable and satisfy $\int_t (u - v) dx = 0$ for all $t \in \tau$.
Then,
\begin{equation}
        \|u - v\|_{L^2\left(\Omega\right)} \leq \ell \|\nabla u - \nabla v\|_{L^2\left(\Omega\right)}. 
\end{equation}
\end{corollary}

This means that conditional on averages over domains of diameter $\ell$, perturbations of $u$ with $L^2$-norm $\delta$ will necessarily lead to perturbations of $\nabla u$ with $L^2$-norm $\delta / \ell$. 
As $\ell$ decreases, the quadratic leading-order term thus becomes the dominant contribution to the loglikelihood. Accordingly, the conditional distribution becomes increasingly Gaussian. (To control the contribution of the nonlinearity, we need to bound the $L^4$ norm, as well as the $L^2$ norm. For $d \leq 3$, Ladyzhenskaya's inequality allows us to bound the $L^4$ norm in terms of the $L^2$ norm, using the fact that the Sobolev norm $H^1$ norm of $u$ is small, with high probability.)

The Cahn-Hilliard process is too rough for point-wise measurements to be defined, which is why we have used the conditioning on averages of scale $\ell \approx \ell^k$ as a substitute for conditioning on the first $k$ elements in the maximin ordering introduced in Section~\ref{sec:maximin}.
If the order $r$ of the elliptic PDE is larger than the spatial dimension $d$, and thus pointwise measurements are well-defined, estimates similar to \citet[][Lemma 1 and Theorem 1]{madych1985estimate} can be used instead of the Poincar{\'e} inequality to obtain similar results when conditioning on subsampled data. (See \citealp{owhadi2017universal,Schafer2017} for examples in the Gaussian case.)

\section{Gibbs sampler for Dirichlet process mixture model\label{app:dpmgibbs}}

The model in Section \ref{sec:dpm} can be fitted using a Gibbs sampler with some Metropolis-Hasting steps. 

This requires additional notation. We introduce cluster indicators $\bc_i = (c_i^{(1)},\ldots,c_i^{(n)})$, such that $c_i^{(j)}$ indicates the cluster to which $\epsilon_i^{(j)}$ belongs. We denote by $n_i$ the number of clusters (i.e., the number of unique entries of $\bc_i$).
Further, let $\tilde\bfmu_i$ and $\tilde\bd_i$ contain the $n_i$ unique cluster-specific parameters in $(\mu_i^{(1)},\ldots,\mu_i^{(n)})$ and $(d_i^{(1)},\ldots,d_i^{(n)})$, respectively; for example, we have $\mu_i^{(j)} = \tilde\bfmu_{i,c_i^{(j)}}$.

For each cluster $k=1,\ldots,n_i$, consider the within-cluster index set $\mathcal{J}_{i,k} = \{l: c_i^{(l)} = k\}$, the cluster size $n_{i,k} = |\mathcal{J}_{i,k}|$, and the cluster average $\bar\epsilon_{i,k} = (1/n_{i,k})\sum_{l \in \mathcal{J}_{i,k}} \epsilon_i^{(l)}$. 
Then, define
$\tilde\xi_{i,k} = (\eta_i \xi_i + n_{i,k} \bar\epsilon_{i,k})/\tilde\eta_{i,k}$,
$\tilde\eta_{i,k}= \eta_{i} + n_{i,k}$,
$\tilde\alpha_{i,k}= \alpha_{i} + n_{i,k}/2$, and
$\tilde\beta_{i,k} = \beta_i + (1/2) \sum_{l \in \mathcal{J}_{i,k}} (\epsilon_i^{(l)} - \bar\epsilon_{i,k})^2 + (1/2) n_{i,k} \eta_i (\bar\epsilon_{i,k} -\xi_i)^2/\tilde\eta_{i,k}$.
Further, let $\tilde\xi_{i,k}^{(-j)}$, $\tilde\eta_{i,k}^{(-j)}$, $\tilde\alpha_{i,k}^{(-j)}$, and $\tilde\beta_{i,k}^{(-j)}$ be the corresponding quantities computed without $\epsilon_i^{(j)}$.

We propose a Markov chain Monte Carlo (MCMC) procedure that, after initialization, cycles through the following steps for a large number of iterations:
\begin{enumerate}
    \item For $i=1,\ldots,N$, sample $\bfepsilon_i | \bY, \tilde\bfmu_i,\tilde\bd_i,\bc_i,\bftheta$ from $\normal_n(\hat\bfepsilon_i,\bD_i-\bD_i\bG_i^{-1}\bD_i)$, where $\hat\bfepsilon_i = \bfmu_i + \bD_i \bG_i^{-1} (\by_i-\bfmu_i)$, $\bG_i = \bC_i + \bD_i$, $\bD_i = \diag((d_i^{(1)})^2,\ldots,(d_i^{(n)})^2)$, and $\bC_i =C_i(\bY_{1:i-1},\bY_{1:i-1})$ depends on $\bftheta$.
    \item For $i=1,\ldots,N$, sequentially sample $c_i^{(1)},\ldots,c_i^{(n)}$ from $p(\bc_i | \bfepsilon_i,\bftheta) \propto \prod_{j=1}^n p(c_i^{(j)}| \bc_i^{(-j)},\bfepsilon_i,\bftheta)$ with $\tilde\bfmu_i,\tilde\bd_i$ integrated out \citep{MacEachern1994,Neal2000}. Specifically, we have $P(c_i^{(j)} = k | \bc_i^{(-j)},\bfepsilon_i,\bftheta) \propto \frac{n^{(j)}_{i,k}}{n-1+\zeta_i} b_{i,k}^{(j)}$ for $k=1,\ldots,n_i+1$, where $n_{i,n_i+1} = \zeta_i$ and $b_{i,k}^{(j)} = p(\epsilon_i^{(j)}| \bfepsilon_i^{(-j)},\bftheta) = t_{2\tilde\alpha_{i,k}^{(-j)}}\big(\tilde\xi_{i,k}^{(-j)},\frac{\tilde\beta_{i,k}^{(-j)}(\tilde\eta_{i,k}^{(-j)}+1)}{\tilde\alpha_{i,k}^{(-j)} \tilde\eta_{i,k}^{(-j)}}\big)$ is the density of a nonstandardized $t$ distribution.
    \item For $i=1,\ldots,N$, sample $(\tilde\mu_{i,k},\tilde d_{i,k}^2)$ from $p(\tilde\bfmu_i,\tilde\bd_i^2| \bfepsilon_i, \bc_i,\bftheta) = \prod_{k=1}^{n_i} \nig(\tilde\mu_{i,k},\tilde d_{i,k}^2|\tilde\xi_{i,k},\tilde\eta_{i,k}, \tilde\alpha_{i,k},\tilde\beta_{i,k})$.
    \item Using Metropolis-Hastings, sample $\bftheta$ from 
    \begin{align*}
    & p(\bftheta|\{\by_i,\bc_i,\bfepsilon_i,\tilde\bfmu_i,\tilde\bd_i: i=1,\ldots,N\}) \allowbreak  
    = p(\bftheta_\zeta|\{\bc_i\}) \allowbreak p(\bftheta_\sigma,\theta_\gamma,\theta_q|\{\by_i,\bfepsilon_i\}) \allowbreak  p(\bftheta_d,\bftheta_\eta|\{\tilde\bfmu_i,\tilde\bd_i \}) \allowbreak \\
    & \quad \textstyle \propto \big( \prod_{i=1}^N \zeta_i^{n_i} \Gamma(\zeta_i)/\Gamma(n+\zeta_i)\big) \big( \prod_{i=1}^N \normal_n(\by_i - \bfepsilon_i | \bfzero,\bC_i) \big) \big( \prod_{i=1}^N \prod_{j=1}^{n_i} \nig(\tilde\mu_{i,k} \tilde d^2_{i,k} | \xi_i, \eta_i, \alpha_i, \beta_i) \big),
    \end{align*}
    where each of the product terms depend on different components of $\bftheta$, which we thus sample and accept/reject separately.
\end{enumerate}
Steps 1--3 can be carried out in parallel for $i=1,\ldots,N$; for Step 4, the terms for $i=1,\ldots,N$ can also be computed in parallel and then combined to obtain the acceptance probabilities for new values of $\bftheta$.

The posterior predictive distribution for a new observation $\by^\star$ can be written as
\[
\textstyle p(\by^\star|\bY) = \prod_{i=1}^N p(y_i^\star|\by_{1:i-1}^\star,\bY),
\]
for which each $p(y_i^\star|\by_{1:i-1}^\star,\bY)$ is approximated as a mixture of Gaussians based on the MCMC output from above. More precisely, given $L$ (thinned) samples of the state variables from the Gibbs sampler, we have
\[
\textstyle
p(y_i^\star|\by_{1:i-1}^\star,\bY) = (1/L) \sum_{l=1}^L \sum_{k=1}^{n_i+1} w_{i,k}^{(l)} \normal(y_i^\star| f_i^{(l)}(\by^\star_{1:i-1}) + \tilde\mu_{i,k}^{(l)}, v_i^{(l)}(\by^\star_{1:i-1})+ (\tilde d_{i,k}^{(l)})^2 ),
\]
where $w_{i,k}^{(l)} = n_{i,k}^{(l)}/(n_i^{(l)} + \zeta_i^{(l)})$,
$f_i^{(l)}(\by^\star_{1:i-1}) = C_i^{(l)}(\by^\star_{1:i-1},\bY_{1:i-1})(\bG_i^{(l)})^{-1}\by_i$, $v_i^{(l)}(\by^\star_{1:i-1}) = C_i^{(l)}(\by^\star_{1:i-1},\by^\star_{1:i-1}) \allowbreak - \allowbreak C_i^{(l)}(\by^\star_{1:i-1},\bY_{1:i-1}) \allowbreak (\bG_i^{(l)})^{-1} \allowbreak  C_i^{(l)}(\bY_{1:i-1},\by^\star_{1:i-1})$, and $(\tilde\mu_{i,n_i+1}^{(l)}, (\tilde d_{i,n_i+1}^{(l)})^2) \sim \nig(\xi_i^{(l)}, \eta_i^{(l)}, \alpha_i^{(l)}, \beta_i^{(l)})$.

\section{Additional simulation study\label{app:simprod}}

Adding to the simulation study in Section \ref{sec:simstudy}, we considered a simulation scenario that deviated more strongly from a simple Mat\'ern GP. Specifically, we simulated the data as a product of a GP as in LR900 (exponential covariance on a regular $30 \times 30$ grid) and the (deterministic) function $\sin((s_1+s_2-1)/.05)$, where $(s_1,s_2)$ are the $(x,y)$-coordinates corresponding to a location or grid point $\bs$. The resulting realizations (Figure \ref{fig:simprod_illus}) mimic that of a process with diagonal advection.

\begin{figure}
\centering
~
\hfill	
	\begin{subfigure}{.3\textwidth}
	\centering
  	\includegraphics[trim=0mm 0mm 10mm 0mm, clip, width =.99\linewidth]{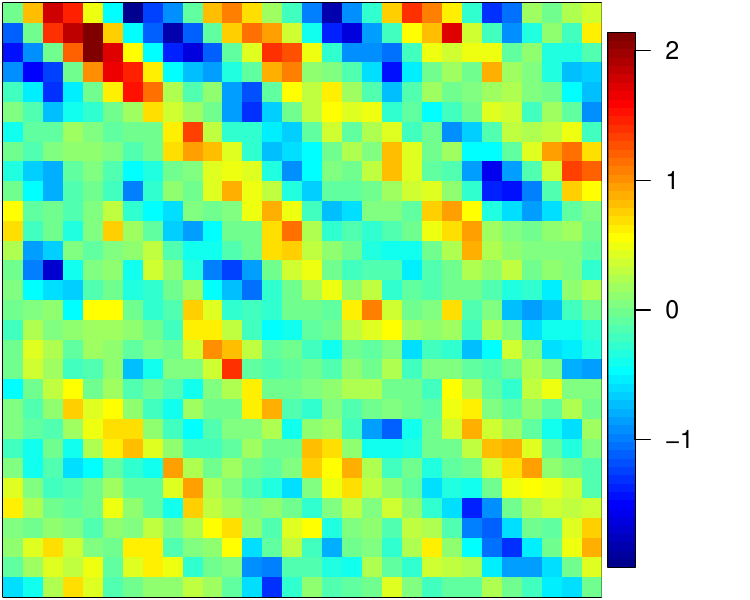}
	\caption{Simulated data}
	\label{fig:simprod_illus}
	\end{subfigure}%
\hfill	
	\begin{subfigure}{.26\textwidth}
	\centering
 	\includegraphics[width =.99\linewidth]{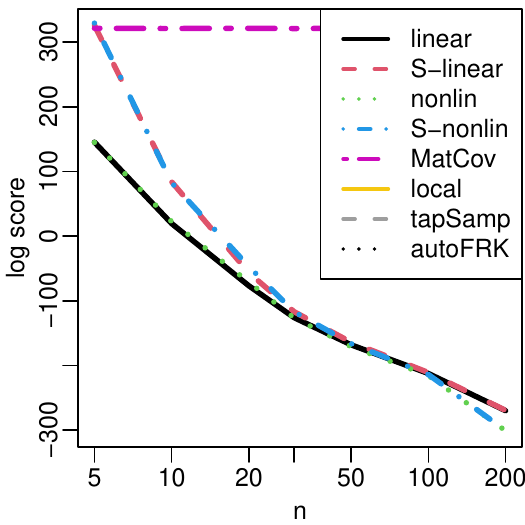}
	\caption{Log-scores}
	\label{fig:ls_prod}
	\end{subfigure}%
\hfill 
~
	\caption{
Results for simulated data based on product of exponential GP and sine function. In (b), \texttt{local}, \texttt{tapSamp}, and \texttt{autoFRK} are not shown because their log-scores were too high.
}
\label{fig:simprod}
\end{figure}

Figure \ref{fig:ls_prod} shows the resulting average log-scores (lower is better) for various training ensemble sizes $n$. \texttt{nonlin} was as or more accurate than all other methods for all $n$.

\section{Spatial prediction\label{app:pred}}

Our method can in principle be used for spatial prediction, but the resulting predictions are unlikely to be useful at spatial locations that are completely unobserved, because the posterior predictive distribution at those locations would simply be the prior predictive distribution. (For simplicity, we assume in this section that $\bftheta$ is known or has been estimated.) As mentioned in Section \ref{sec:conclusions}, our method could be modified to shrink toward a specific parametric covariance and hence to produce more meaningful prior predictive distributions at unobserved locations, which could be accomplished by extending the approach in \citet[][Sect.~R1--R2]{Kidd2020}.

However, our current method can produce competitive predictions at locations that are partially observed. Assume that the $n \times N$ data matrix $\bY$ is fully observed, and we also have a vector $\by^*$ that is partially observed. We would like to predict the missing entries of $\by^*$. The most straightforward way to do this is to order the partially unobserved spatial locations last \citep{Katzfuss2018,Schafer2020}, and so we assume that $\by^*_{1:N_o}$ is observed and $\by^*_{N_o+1:N}$ is missing. The posterior predictive distribution can then be seen from the results in Section \ref{app:proofs} to be
\[
\textstyle p( \by^*_{(N_o+1):N} |\by^*_{1:N_o},\bY) = \prod_{i=N_o+1}^N p(y_i^*|\by^*_{1:i-1},\bY) = \prod_{i=N_o+1}^N t_{2\tilde\alpha_i}\big(y_i^\star\big|\hat f_i(\by^\star_{1:i-1}), \hat d_i^2( v_i(\by^\star_{1:i-1})+1 )\big),
\]
where the $t$ distributions are exactly as in \eqref{eq:tstar}.

To gauge the accuracy of predictions at partially observed locations, we carried out a comparison on the Americas climate data similar to the comparison in Section \ref{sec:application} and Figure \ref{fig:pre4}. We again assumed that $n$ replicates were (fully) observed, but now we assumed that for held-out test data, the first half (i.e., $N_o = N/2 = 1{,}369$) locations were observed, and the remaining $N/2$ locations were unobserved and to be predicted. We evaluated the accuracy of the joint predictive distribution for the held-out test data using the log-score. Interestingly, this is equivalent to simply considering only some of the summands in the sum making up the log-score for the full distribution considered in Figure \ref{fig:pre4}. 
Specifically, the log-score in the partially observed setting is
$\log p( \by^*_{(N_o+1):N} |\by^*_{1:N_o},\bY) = \sum_{i=N_o+1}^N \log p(y_i^*|\by^*_{1:i-1},\bY)$, while the log-score for the entire distribution (i.e., for Figure \ref{fig:pre4}) is
$\log p( \by^* |\bY) = \sum_{i=1}^N \log p(y_i^*|\by^*_{1:i-1},\bY)$.
Thus it is not surprising that the relative results for the partially observed setting in Figure \ref{fig:precipcomp_ho} were similar to those in Figure \ref{fig:pre4}.

\begin{figure}
\centering
\includegraphics[width =.3\linewidth]{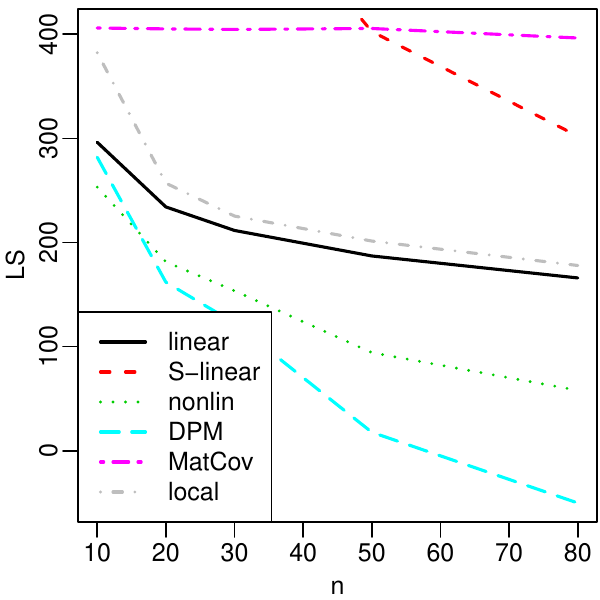}
\caption{
For precipitation anomalies, comparison of log-score (LS) for estimated joint distribution on a hold-out set consisting of half the locations, as a function of ensemble size $n$. As in Figure \ref{fig:pre4}, several competing methods are not shown because their LS were too high.}
\label{fig:precipcomp_ho}
\end{figure}

\section{Dimension reduction and reconstruction\label{app:recon}}

As mentioned in Section \ref{sec:spatialinference}, our method can be employed as a nonlinear spatial version of principal component analysis (PCA), which is commonly used for dimension reduction. In the environmental and climate sciences, principal components are highly popular and referred to as empirical orthogonal functions \citep[e.g., see][for a review]{Hannachi2007}. 
We compared the dimension-reduction and reconstruction accuracy of our method to regular (linear) PCA on the Americas climate data described in Section \ref{sec:application}. The methods were both fitted on $n$ fully observed training samples and were then allowed to extract $k$ numbers from each test sample or image; PCA stored the first $k$ PC scores, while our method stored the first $k$ map coefficients $\bz^*_{1:k}$ (see Section \ref{sec:spatialinference}). 

Based on the $k$ numbers, the task was to provide a probabilistic prediction of the reconstructed test field, whose quality was evaluated via the log-score. For PCA, the distribution was taken to be Gaussian whose mean was equal to the point prediction and whose covariance matrix was $\tau^2\bI_N$, where $\tau^2$ was the average squared difference between the point prediction and the test field. We considered a range of values for $n$. We set $k=n$, because it is not possible to estimate more than $n$ PCs from $n$ training samples. (Note that our method does not suffer from this limitation.)
As shown in Figure \ref{fig:precipcomp_recon}, our method strongly outperformed PCA. As an added advantage, our method only needs to access $n$ data points (the first $n$ in maximin ordering) from each test image to compute the first $n$ map coefficients, whereas PCA needs access to the full image to compute PC scores.

\begin{figure}
\centering
\includegraphics[width =.3\linewidth]{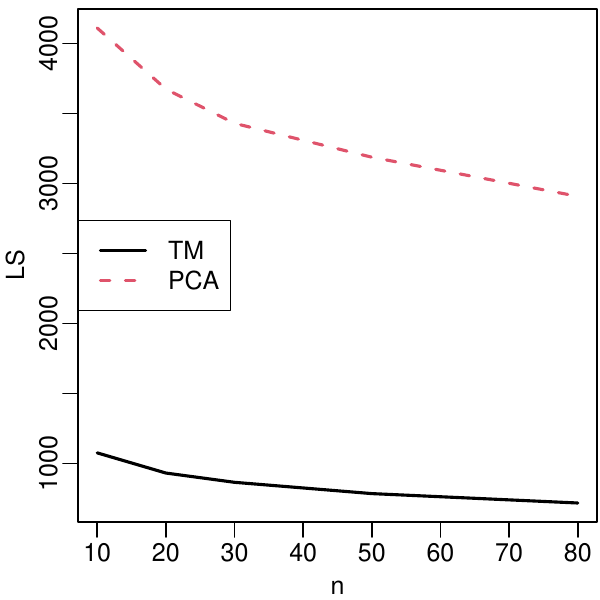}
\caption{
Reconstruction accuracy for test images in terms of log-score (LS) based on $n$ training samples and reducing the test images to $k=n$ principal components, for our nonlinear transport map (TM) versus standard principal component analysis (PCA).}
\label{fig:precipcomp_recon}
\end{figure}

\section{Comparison to VAE\label{app:vae}}

As stated in Section \ref{sec:intro}, our approach
can be viewed as a Bayesian shallow autoencoder, with the posterior transport map and its inverse acting as the encoder and decoder, respectively. We implemented a (deep) VAE \citep{Kingma2014} as a comparison method. 
We used four 2D-convolutional layers for both the encoder and the decoder, summing up to a total of 150,723 model parameters for the Americas climate data. For optimization, we used the Adam optimizer with a learning rate of $10^{-4}$ and a cosine-annealing scheduler. The training stopped after 500 epochs.

We compared samples from the fitted VAE and transport map, as shown in Figure \ref{fig:vaecomp_GP} for a GP with exponential covariance and in Figure \ref{fig:vaecomp} for the climate application, all based on training ensemble sizes around 100. 
Plots within each row are not meant to be similar to each other. Instead, for each figure, each individual panel should be an independent sample from the same distribution. Thus, the quality of the methods should be assessed via a kind of Turing test: If one randomly permuted, for example, the plots in the first two columns of each figure, we would argue that it would be difficult to tell which samples were from the exact distribution and which were from the transport map.
In contrast, the VAE samples are clearly different; the variance of the VAE samples is too small, and the spatial features are heavily blurred, which is a known issue for VAEs \citep[e.g.,][Ch.~20.10.3]{Goodfellow2016}.

We also attempted to include a comparison to a GAN designed for climate-model output \citep{Besombes2021}, but we were not able to obtain useful results for the small sample sizes considered here.

\begin{figure}
\centering
\includegraphics[width =.99\linewidth]{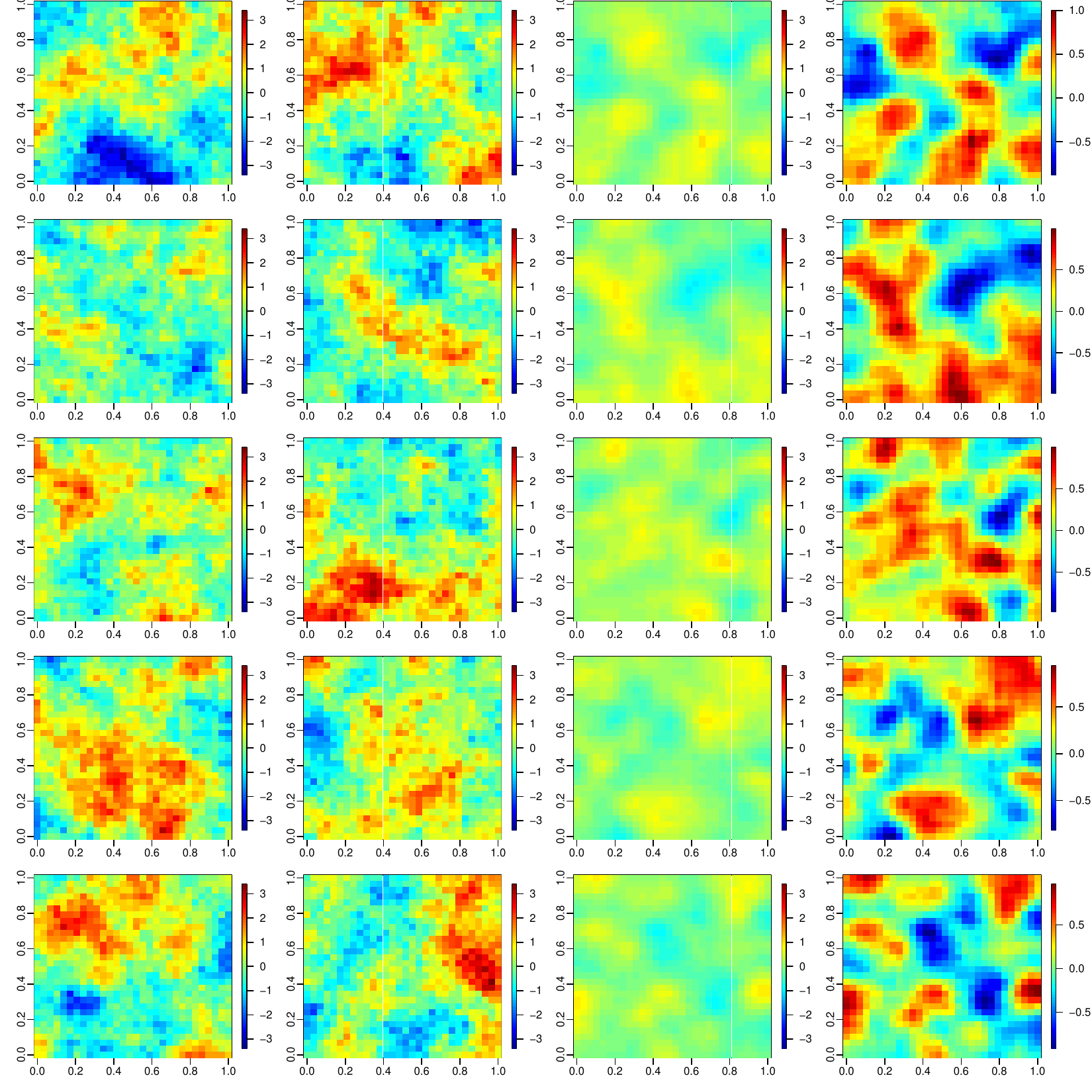}
\caption{For simulated data in the LR900 scenario described in Section \ref{sec:simstudy}: Left (first) column: Five samples from the exact GP model. Second: Samples from our fitted transport-map model (\texttt{nonlin}). Third: Samples from a fitted VAE. Fourth: Same VAE samples but on a different color scale with values roughly between -1 and +1. (The first three columns are on the same color scale.)}
\label{fig:vaecomp_GP}
\end{figure}

\begin{figure}
\centering
\includegraphics[width =.99\linewidth]{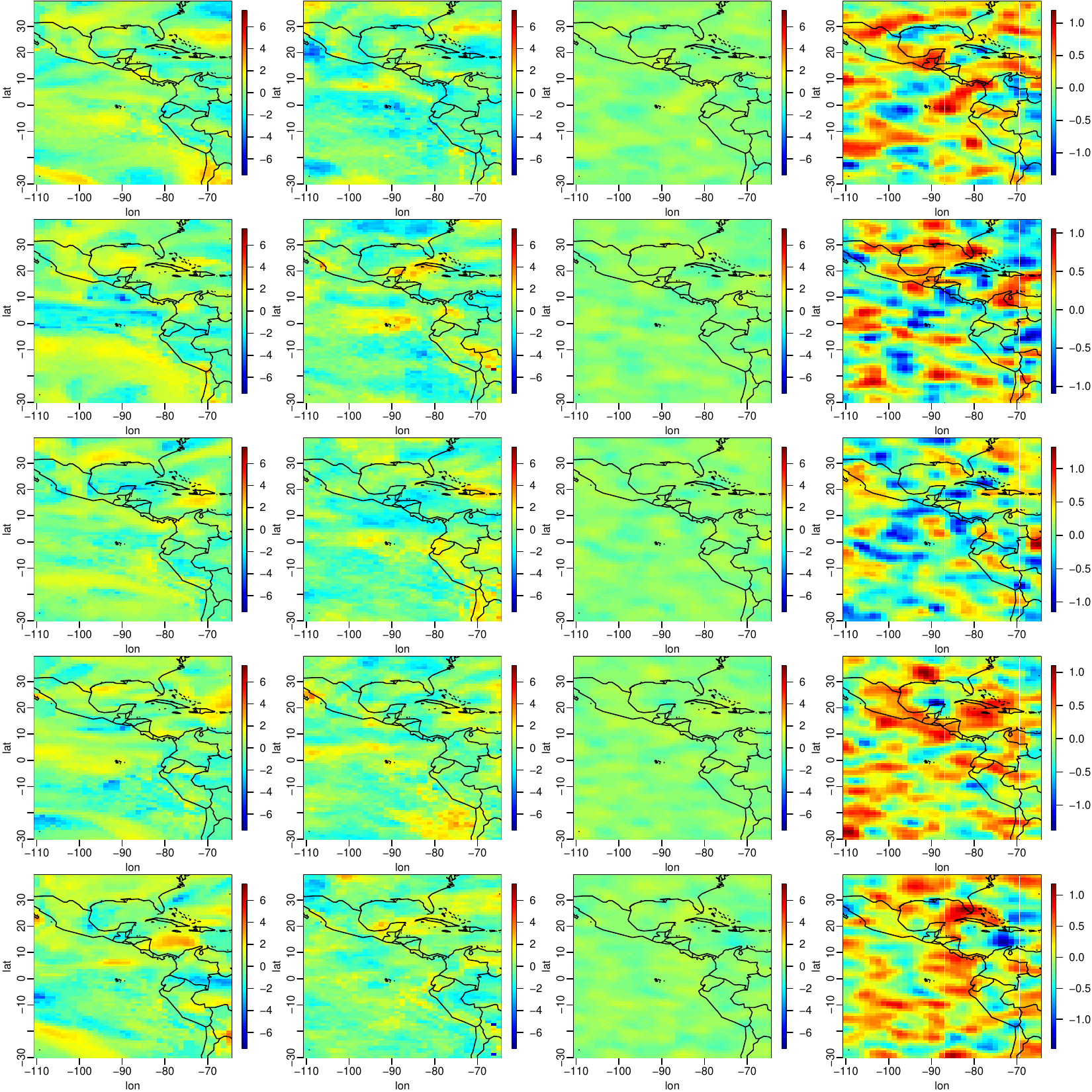}
\caption{For precipitation anomalies in the Americas subregion produced by a climate model (see Section \ref{sec:application}): Left (first) column: Five samples from the climate model. Second: Samples from our fitted transport-map model (\texttt{nonlin}). Third: Samples from a fitted VAE. Fourth: Same VAE samples but on a different color scale with values roughly between -1.5 and +1.5. (The first three columns are on the same color scale.)}
\label{fig:vaecomp}
\end{figure}

\bibliographystyle{apalike}
\bibliography{mendeley,additionalrefs}

\end{document}